\documentclass[11pt,tightenlines,aps,prd,notitlepage,nofootinbib,superscriptaddress]{revtex4-1}
\usepackage[utf8]{inputenc}
\usepackage[table,dvipsnames]{ xcolor}

\usepackage{multirow}
\usepackage{caption}
\usepackage{subcaption}
\usepackage{hyperref}
\usepackage{graphicx}
\usepackage{mathtools}
\usepackage{amsmath,amssymb,amsfonts,amsthm,latexsym,stmaryrd,physics,mathrsfs}
\numberwithin{figure}{section}
\allowdisplaybreaks[4]
\usepackage{marginnote}
\usepackage{color}
\usepackage{bm}
\usepackage{float}
\usepackage{nicefrac}
\usepackage{microtype} 
\usepackage{booktabs}
\usepackage{verbatim}
\usepackage{setspace}
\usepackage[T1]{fontenc}
\usepackage[utf8]{inputenc}
\usepackage{stfloats}
\usepackage{diagbox}
\usepackage{dsfont}
\usepackage{tikz}
\usepackage{geometry}

\hypersetup{
	colorlinks=true
,urlcolor=blue
,anchorcolor=blue
,citecolor=blue
,filecolor=blue
,linkcolor=blue
,menucolor=blue
,pagecolor=blue
,linktocpage=true
,pdfproducer=medialab
}

\def\beq{\begin{equation}}
\def\eeq{\end{equation}}
\newcommand{\bea}{\begin{eqnarray}}
\newcommand{\eea}{\end{eqnarray}}
\def\bi{\begin{itemize}}
\def\ei{\end{itemize}}
\def\ba{\begin{array}}
\def\ea{\end{array}}
\def\bfig{\begin{figure}}
\def\efig{\end{figure}}

\newtheorem{theorem}{Theorem}[section]

\newtheorem{lemma}[theorem]{Lemma}
\newtheorem*{theorem1}{Main Result}

\def\be{\begin{eqnarray}}
\def\ee{\end{eqnarray}}

\usepackage{geometry}
\begin{document}

\title{Beyond Expectation Values: Generalized Semiclassical Expansions for Matrix Elements of Gauge Coherent States
}

\author{Haida Li}
\email{eqwaplay@scut.edu.cn}
\affiliation{School of Physics and Optoelectronics, South China University of Technology, Guangzhou 510641, China}
\affiliation{Institute for Theoretical Sciences and Department of Physics, Westlake University, Hangzhou 310024, Zhejiang, China}

\author{Hongguang Liu} 
\email{Corresponding author: liuhongguang@westlake.edu.cn}
\affiliation{Institute for Theoretical Sciences and Department of Physics, Westlake University, Hangzhou 310024, Zhejiang, China}

\begin{abstract}

We derive an asymptotic expansion for off-diagonal coherent-state matrix elements of non-polynomial operators in gauge theories admitting holomorphic coherent-state representations. The derivation combines stationary-phase analysis with an operator-level treatment of the Taylor remainder, and yields explicit semiclassical error control under stated assumptions. As a primary application, we formulate the expansion for volume and flux related operators in Loop Quantum Gravity and compare it with the standard diagonal expansion proposed in \cite{Giesel:2006um}. By organizing the expansion around the genuine off-diagonal Berezin symbol rather than a diagonal expectation value, the resulting formula preserves the full holomorphic structure of the geometric phase and reproduces benchmark matrix elements accurately in the numerical regimes tested here, particularly when the coherent-state labels are well separated.

\end{abstract}

\maketitle

\section{Introduction}

Coherent states provide a semiclassical representation in which quantum dynamics can be related directly to classical phase-space data \cite{Schrodinger:1926zz}. In gauge theories, they are particularly useful because they furnish both an overcomplete resolution of identity and a natural notion of off-diagonal Berezin symbols \cite{klauder1985coherent,gerry2023introductory,Berezin:1974du}. The present work is concerned with the computation of such off-diagonal matrix elements for non-polynomial operators.

Based on the coherent state framework, the coherent state path integral \cite{negele2018quantum,nagaosa2013quantum,dupuis2023field,Kochetov:2018uhs} allows for the description of quantum dynamics directly on the classical phase space. The quantum amplitude is recast as a sum over classical trajectories, where the weighting is determined by an effective action that captures the underlying quantum commutation relations \cite{Klauder:1979gi,Faddeev:1969su}. The mathematical rigor of this approach is deeply rooted in the Berezin quantization scheme and Geometric Quantization, utilizing Berezin symbols to resolve ordering ambiguities and defining the Kähler polarization via the coherent state overlap \cite{1974IzMat...8.1109B,Berezin:1974du}. It must be emphasized that the rigorous evaluation of the off-diagonal matrix elements between distinct coherent states constitutes the primary motivation for such formalisms. In the coherent-state path integral, one repeatedly encounters short-time matrix elements of the form $\langle z_{k+1}|e^{-i\epsilon \widehat H}|z_k\rangle$, and therefore needs control over genuinely off-diagonal coherent-state symbols. In the strict continuous-time limit, diagonal expectation values may suffice at a formal level, but this is no longer adequate once one works with finite time steps or wishes to retain the full off-diagonal holomorphic structure. The latter is also relevant for analytic continuation, although that aspect will only serve here as a motivation rather than as a developed application.

The application of these geometric formalisms to gauge theories relies inherently on the construction of gauge group coherent states. A prominent foundational example is the Hall coherent state developed for Yang-Mills (YM) theories \cite{HALL1994103}, which was subsequently generalized into the Thiemann complexifier coherent state \cite{Thiemann:2000bw,Thiemann:2000ca,Thiemann:2000bx,Thiemann:2000by} for the study of Loop Quantum Gravity (LQG) \cite{Thiemann:2007pyv,Rovelli:2014ssa,Ashtekar:2017yom,Perez:2012wv}. By employing the reduced phase space quantization scheme \cite{giesel2010algebraic,thiemann2006reduced,Brown:1994py,Husain:2011tk}, canonical quantum gravity can be recast into a deparametrized form that is structurally analogous to standard non-perturbative gauge field theories. In this framework, the physical time evolution is well-defined, allowing for the direct implementation of the coherent state path integral to evaluate transition amplitudes in quantum gravity \cite{Han:2019vpw,Han:2019feb,Han:2020iwk,Han:2020chr,Bodendorfer:2020ovt,Han:2020uhb,Han:2021cwb}.

In background-independent gauge theories such as LQG, the relevant geometric observables and Hamiltonians are typically non-polynomial functions of the basic holonomy and flux variables. This is not a matter of convenience but a structural feature of theory: for example, the volume operator \cite{ Rovelli:1994ge,Ashtekar:1997fb,Brunnemann:2004xi} is defined through fractional powers of polynomial flux combinations. As a result, exact coherent-state matrix elements of these operators are generally difficult to obtain, even when the action of the underlying polynomial operators is explicitly known \cite{Brunnemann:2007ca,DePietri:1996tvo}.

To formulate a valid coherent state path integral, one must explicitly compute the off-diagonal matrix elements $\langle z | f(\hat{A}) | z' \rangle$ of these non-polynomial operators between distinct states. Historically, semiclassical treatments in LQG often rely on expansions organized around diagonal expectation values, most notably in the Giesel-Thiemann framework \cite{Giesel:2006um}. Such formulas are adequate for the strict continuous-time limit, where neighboring coherent states are infinitesimally close. However, once one keeps a finite discretization or wishes to preserve the full off-diagonal holomorphic structure, a purely diagonal expansion becomes structurally restrictive. The issue is not simply the size of the correction, but the choice of the expansion center itself: a diagonal expectation value does not retain the full geometric information carried by the genuine off-diagonal Berezin symbol.

To resolve this problem, in this work we derive an asymptotic expansion for off-diagonal coherent-state matrix elements of operators of the form $(\widehat A^2)^q$, organized around the off-diagonal Berezin symbol $C(z,z')=\langle z|\widehat A|z'\rangle/\langle z|z'\rangle$. The derivation combines a stationary-phase analysis of the coherent-state overlap with an operator-level treatment of the Taylor remainder, which together yield an explicit semiclassical error estimate under stated assumptions. We then apply the formalism to the volume-related operator sector of LQG and compare the resulting expansion both with the standard diagonal-centered formula and with independently computed numerical benchmark data.

The paper is organized as follows. In Sec. \ref{sectionII} we first review the coherent-state framework and formulate the off-diagonal Berezin-symbol problem. We then establish the key semiclassical estimates for shifted operator powers and prove the main asymptotic expansion formula. In Sec. \ref{sectionIV} we apply the general result to LQG coherent states, verify the required assumptions for the relevant flux and volume operators, and compare the new formula with both the standard LQG expansion and numerical benchmark calculations. We conclude in Sec. \ref{sectionV} with a discussion and outlook.

\section{General Expansion Formula}\label{sectionII}

We now set up the coherent-state notation needed for the asymptotic expansion. The central object is the off-diagonal Berezin symbol of a non-polynomial operator built from a polynomial operator $\widehat A$. Our goal in this section is to isolate the semiclassical mechanism that controls such symbols and to formulate the corresponding expansion formula.

\subsection{Berezin symbol and general coherent states}

To fix notation, we begin with canonical coherent states. The subsequent formulas will then be stated in a form suitable for more general coherent-state families with a Kähler-type overlap structure.

We use the convention $[a,a^\dagger]=1$ and label coherent states by a classical variable $z=\left(z_1, \ldots, z_n\right)$ such that
\begin{equation}
a|z\rangle=\frac{z}{\sqrt{\hbar}}|z\rangle.
\end{equation}
Here $z$ carries the dimension of a classical phase space coordinate. 
Then
\begin{equation}\label{defz}
|z\rangle=\exp\left(\frac{z a^\dagger-\bar z a}{\sqrt{\hbar}}\right)|0\rangle
=\exp\left(-\frac{|z|^2}{2\hbar}\right)\exp\left(\frac{z a^\dagger}{\sqrt{\hbar}}\right)|0\rangle,
\end{equation}
and therefore
\begin{equation}\label{prop1}
\langle z'|z\rangle
=\exp\left(\frac{\bar z' z-\frac12|z'|^2-\frac12|z|^2}{\hbar}\right).
\end{equation}

More generally, for a coherent-state family with Kähler-type overlap, we write
\begin{equation}
    \langle z'|z\rangle
=\exp\left(\frac{K(z,\bar z')-\frac12K(z,\bar z)-\frac12K(z',\bar z')}{\hbar}\right),
\end{equation}
where $K(z,\bar{z}')$ is the Kähler potential. In this simple case, the Kähler potential is the bilinear form:
\begin{equation}\label{Kpotential}
    K(z,\bar{z}')\equiv\bar{z}^{\prime} z,
\end{equation}
and $\frac{1}{2}K(z,\bar{z})$ and $\frac{1}{2}K(z',\bar{z}')$ come from normalization. For a more general coherent-state family, the Kähler potential $K(z,\bar{z}')$ can be a more complicated function of $z$ and $\bar{z}'$, but it still satisfies the same overlap structure with a corresponding resolution of identity:
\begin{equation}\label{rI}
\mathbb{I}=\int dv(z)|z\rangle\langle z|,
\end{equation}
Here $dv(z)$ denotes the measure associated with the overcomplete coherent-state family.

The expectation value of an operator $\widehat{O}$ in a coherent state $|z\rangle$ is:
\begin{equation}
\frac{\langle z| \hat{O}|z\rangle}{\langle z| z\rangle}=O_{\operatorname{diag}}(z, \bar{z}),
\end{equation}
which is a real-valued function of $z$ and its complex conjugate $\bar{z}$. Since coherent states are constructed to be holomorphic in $z$ and anti-holomorphic in $\bar{z}$, one can often analytically continue $O_{\operatorname{diag}}(z, \bar{z})$ to a function of two independent complex variables:
\begin{equation}
    O(z, z')=\frac{\left\langle z\right| \hat{O}|z^{\prime}\rangle}{\langle z|z^{\prime}\rangle}.
\end{equation}
This function $O(z, z')$ is called the Berezin symbol of the operator $\widehat{O}$, which is holomorphic in $z'$ and anti-holomorphic in $z$. Note that for the Berezin symbol we do not need to use normalized coherent states.

 In a quantized gauge field theory, the physical dynamics at the semi-classical level are manifested by the computation of the transition amplitude between initial and final coherent states:
\begin{equation}\label{transA}
\left\langle z_f\right| e^{-i \hat{H} t}\left|z_i\right\rangle.
\end{equation}

To compute the transition amplitude (\ref{transA}), usually we need to adopt the coherent state path integral formalism by considering the time slice to divide $t$ into $N$ intervals: $\epsilon:=\frac{t}{N}$. Inserting the resolution of identities into (\ref{transA}), we have:
\begin{equation}\label{cpath}
\left\langle z_f\right| e^{-i \hat{H} t}\left|z_i\right\rangle=\int \prod_{k=1}^{N-1} d v(z_k)\left\langle z_N\right| e^{-i \epsilon \hat{H}}\left|z_{N-1}\right\rangle \cdots\left\langle z_1\right| e^{-i \epsilon \hat{H}}\left|z_0\right\rangle,
\end{equation}
where $z_N=z_f, z_0=z_i$. Each small slice (for small $\epsilon$) can be estimated as:
\begin{equation}
\left\langle z_{k+1}\right| e^{-i \epsilon \hat{H}}\left|z_k\right\rangle \approx\left\langle z_{k+1}\right|(1-i \epsilon \hat{H})\left|z_k\right\rangle=\left\langle z_{k+1} \mid z_k\right\rangle-i \epsilon\left\langle z_{k+1}\right| \hat{H}\left|z_k\right\rangle.
\end{equation}
This elementary time-slicing step already shows why off-diagonal symbols are unavoidable: each short-time factor is evaluated between neighboring but distinct coherent states. The problem addressed in this paper is therefore not an optional refinement of the diagonal formalism, but part of the basic discrete coherent-state path-integral construction.

To manifest the general microlocal machinery mentioned above in a physically relevant context, we focus on a prototypical non-polynomial structure ubiquitous in LQG (e.g., associated with the quantum volume operator): the fractional power operator $(\widehat{A}^2)^q$ with $0 < q \leq \frac{1}{4}$. Rather than attempting a direct evaluation, the target operator $(\widehat{A}^2)^q$ can be formally expanded by inserting the overcomplete resolution of identity (\ref{rI}):
\begin{equation}\label{zexpand111}
\begin{split}
(\widehat{A}^2)^q&=\int dv(z_1)dv(z_2)\langle z_1 |(\widehat{A}^2)^q|z_2 \rangle| z_1\rangle\langle z_2|\\
&=\int dv(z_1)dv(z_2)[(\widehat{A}^2)^q](z_1 ,z_2)\langle z_1|z_2\rangle | z_1\rangle\langle z_2|,
\end{split}
\end{equation}
where $[(\widehat{A}^2)^q](z_1,z_2)$ is the analytically continued Berezin symbol of the specific operator. The same strategy will later be seen to extend to a broader analytic class $f(\widehat{A})$, but the derivation is most transparent in this representative case. 

\subsection{ Asymptotic estimate for shifted operator powers}\label{sectionIII}
In the applications below, the coherent-state labels are complex. The stationary-phase argument is used locally, after resolving the complex labels into their underlying real coordinates near the relevant non-degenerate saddle.
We then introduce the theorem of saddle point expansion \cite{hormander2015analysis}:
\begin{theorem}[Hörmander’s theorem 7.7.5]\label{theorem111}
Let K be a compact subset in $\mathbb{R}^n$, X an open neighborhood of K, and k a positive integer. If: (1) the complex functions $u \in C_0^{2 k}(K)$, $f \in C^{3 k+1}(X)$ and $\operatorname{Im} f(x) \geq 0 \quad\forall x\in X$, (2) there is a unique point $x_0 \in K$ satisfying $\operatorname{Im}\left(S\left(x_0\right)\right)=0$, $f^{\prime}\left(x_0\right)=0$, $\operatorname{det}\left(f^{\prime \prime}\left(x_0\right)\right) \neq 0$ (where $f''$ denotes the Hessian matrix), and $f^{\prime}(x) \neq 0 \quad\forall x\in K \backslash\left\{x_0\right\}$, then we have the following estimation:
\begin{equation}\label{saddle1}
\left|\int_K u(x) e^{i \lambda f(x)} d x-e^{i \lambda f \left(x_0\right)}\left[\operatorname{det}\left(\frac{\lambda f^{\prime \prime}\left(x_0\right)}{2 \pi i}\right)\right]^{-\frac{1}{2}} \sum_{s=0}^{k-1}\left(\frac{1}{\lambda}\right)^s L_s u\left(x_0\right)\right| \leq C\left(\frac{1}{\lambda}\right)^k \sum_{|\alpha| \leq 2 k} \sup \left|D^\alpha u\right| .
\end{equation}
Here the constant $C$ is bounded when $f$ stays in a bounded set in $C^{3 k+1}(X)$. We use the standard multi-index notation $\alpha=\left\langle\alpha_1, \ldots, \alpha_n\right\rangle$ and:
\begin{equation}
D^\alpha=(-i)^|\alpha| \frac{\partial^{|\alpha|}}{\partial x_1^{\alpha_1} \ldots \partial x_n^{\alpha_n}}, \quad \text { where } \quad|\alpha|=\sum_{i=1}^n \alpha_i,
\end{equation}
 $L_s u\left(x_0\right)$ is defined as:
\begin{equation}\label{LSU}
L_s u\left(x_0\right)=i^{-s} \sum_{l-m=s} \sum_{2 l \geq 3 m} \frac{(-1)^l 2^{-l}}{l ! m !}\left[\sum_{a, b=1}^n H_{a b}^{-1}\left(x_0\right) \frac{\partial^2}{\partial x_a \partial x_b}\right]^l\left(g_{x_0}^m u\right)\left(x_0\right),
\end{equation}
where $H(x)=f^{\prime \prime}(x)$ denotes the Hessian matrix and the function $g_{x_0}(x)$ is given by:
\begin{equation}
g_{x_0}(x)=f(x)-f\left(x_0\right)-\frac{1}{2} H^{a b}\left(x_0\right)\left(x-x_0\right)_a\left(x-x_0\right)_b,
\end{equation}
satisfying $g_{x_0}\left(x_0\right)=g_{x_0}^{\prime}\left(x_0\right)=g_{x_0}^{\prime \prime}\left(x_0\right)=0$. For each $s$, $L_s$ is a differential operator of order $2s$ acting on $u(x)$.
\end{theorem}

In order to prove this relation, we first prove the following lemma regarding the asymptotic evaluation of intermediate operator insertions:

\begin{lemma}\label{lemma1}
    Let $\widehat{F}$ be a polynomial operator on $\mathcal{H}$, and its Berezin symbol $\frac{\langle z|\widehat{F}|z_i\rangle}{\langle z|z_i\rangle}$ be a smooth, slowly varying function of $z$ and $z_i$ that is independent of the large parameter $1/\hbar$ in the exponent. The following approximation holds:
    \begin{eqnarray}\label{lemma1eq}
        \frac{\langle z|\widehat{F}|z'\rangle }{\langle z|z'\rangle } = \frac{\langle z|\widehat{F}|\tilde{z}\rangle}{\langle z|\tilde{z}\rangle} \left(1 +  \mathcal{O}(\hbar) \right)= \frac{\langle\tilde{z}|\widehat{F}|z'\rangle}{\langle \tilde{z}|z'\rangle} \left(1 +  \mathcal{O}(\hbar) \right),
    \end{eqnarray}
    where $\tilde{z}$ is the isolated single dominant critical point determined strictly by the geometric phase of the coherent state overlap, independent of $\widehat{F}$.
\end{lemma}

The above lemma will be proved directly by a saddle-point analysis of the coherent-state resolution of identity. Before entering the proof, however, it is useful to record a heuristic consistency check with the standard semiclassical picture of Berezin-Toeplitz quantization. In that picture, the intermediate coherent-state label is localized near the geometric saddle selected by the overlap phase, so that the corresponding off-diagonal symbol is evaluated near the classical trajectory connecting the endpoints. This observation is intended only as an interpretation of the lemma and will not be used in the proof below.

What is actually needed in the proof is the following technical point: the amplitude function
\begin{equation}
u_\hbar(z,z_i)\equiv \frac{\langle z|\widehat{F}|z_i\rangle}{\langle z|z_i\rangle}
\end{equation}
is semiclassically slow compared with the exponential phase coming from the coherent-state overlap. More precisely, for polynomial operators $\widehat{F}$ built from the basic holomorphic variables and $\hbar$-weighted differential operators, the derivatives of $u_\hbar$ required in the stationary-phase expansion remain of order $\mathcal{O}(1)$ in the limit $\hbar\to 0$, and $u_\hbar$ contains no additional factors of the form $e^{c/\hbar}$.

To see this, recall that in the holomorphic representation of gauge coherent states on a K\"ahler manifold, the fundamental operators act by coordinate multiplication and by $\hbar$-weighted derivatives, schematically of the form
\begin{equation}
\hat{f}\sim f(z_i)+\hbar \partial_{z_i}.
\end{equation}
The coherent-state overlap takes the form
\begin{equation}
\langle z|z_i\rangle =v(\bar{z},z_i)e^{K(\bar{z}, z_i)/\hbar},
\end{equation}
where $K$ is the K\"ahler potential and $v(\bar{z},z_i)$ is independent of the exponential semiclassical scale. Therefore,
\begin{equation}
\hbar \partial_{z_i} \langle z|z_i\rangle=
v(\bar{z},z_i)\hbar \partial_{z_i} e^{\frac{1}{\hbar}K(\bar{z},z_i)}
+
e^{\frac{1}{\hbar}K(\bar{z},z_i)}\hbar \partial_{z_i}v(\bar{z},z_i)=
\left(\partial_{z_i}K(\bar{z},z_i)\right)e^{\frac{1}{\hbar}K(\bar{z},z_i)}
+\mathcal{O}(\hbar)e^{\frac{1}{\hbar}K(\bar{z},z_i)}.
\end{equation}
The factor $\hbar$ attached to the differential operator cancels the $1/\hbar$ produced by differentiating the exponential. Iterating this argument shows that any polynomial operator $\widehat{F}$ produces only polynomial expressions in derivatives of $K$ and positive powers of $\hbar$, multiplied by the same exponential overlap factor. After division by $\langle z|z_i\rangle$, the latter cancels, and the resulting amplitude $u_\hbar(z,z_i)$ remains a smooth function with derivatives of order $\mathcal{O}(1)$ on the semiclassical scale relevant for the saddle-point analysis.

\begin{proof}
By inserting the resolution of identity
\begin{equation}
\mathbb{I}=\int dv(z_i)\,|z_i\rangle\langle z_i|
\end{equation}
into the matrix element $\langle z|\widehat{F}|z'\rangle$, we obtain
\begin{equation}\label{intF}
\begin{split}
\frac{\langle z|\widehat{F}|z'\rangle}{\langle z|z'\rangle}
&=
\frac{1}{\langle z|z'\rangle}
\int dv(z_i)\,\langle z|z_i\rangle \langle z_i|z'\rangle
\frac{\langle z|\widehat{F}|z_i\rangle}{\langle z|z_i\rangle}\\
&=
\frac{1}{\langle z|z'\rangle}
\int d\tilde{v}(z_i)\,
e^{\frac{1}{\hbar}S(z,z_i,z')}
\frac{\langle z|\widehat{F}|z_i\rangle}{\langle z|z_i\rangle},
\end{split}
\end{equation}
where
\begin{equation}
S(z,z_i,z')=K(z,\bar{z}_i)+K(z_i,\bar{z}')-K(z_i,\bar{z}_i)
\end{equation}
is the geometric phase function associated with the overlap, and the measure factor has been absorbed into $d\tilde{v}(z_i)$.

By the discussion preceding the proof, the amplitude
\begin{equation}
u_\hbar(z_i):=\frac{\langle z|\widehat{F}|z_i\rangle}{\langle z|z_i\rangle}
\end{equation}
has no additional exponential dependence of order $e^{c/\hbar}$ and all derivatives needed in the stationary-phase expansion remain of order $\mathcal{O}(1)$. Therefore the location of the leading saddle is determined by the phase function $S$ alone. Let $\tilde{z}$ denote the isolated non-degenerate critical point satisfying
\begin{equation}
\frac{\partial S(z,z_i,z')}{\partial z_i}=0,
\qquad
\frac{\partial S(z,z_i,z')}{\partial \bar{z}_i}=0.
\end{equation}
Applying the saddle-point expansion of Theorem \ref{theorem111} then gives
\begin{equation}
\langle z|\widehat{F}|z'\rangle=
\frac{1}{\sqrt{\det(-S''(\tilde{z}))/2\pi\hbar}}
e^{\frac{1}{\hbar}S(z,\tilde{z},z')}
\left(
\frac{\langle z|\widehat{F}|\tilde{z}\rangle}{\langle z|\tilde{z}\rangle}
+\mathcal{O}(\hbar)
\right).
\end{equation}
Similarly, the overlap integral yields
\begin{equation}\label{overlapint}
\langle z|z'\rangle=
\int dv(z_i)\,e^{\frac{1}{\hbar}S(z,z_i,z')}=
\frac{1}{\sqrt{\det(-S''(\tilde{z}))/2\pi\hbar}}
e^{\frac{1}{\hbar}S(z,\tilde{z},z')}
\left(
1+\mathcal{O}(\hbar)
\right).
\end{equation}
Dividing the two expressions, the fluctuation determinant and the exponential factor cancel, and we obtain
\begin{equation}
\frac{\langle z|\widehat{F}|z'\rangle}{\langle z|z'\rangle}=
\frac{\langle z|\widehat{F}|\tilde{z}\rangle}{\langle z|\tilde{z}\rangle}
\left(
1+\mathcal{O}(\hbar)
\right).
\end{equation}
The second relation
\begin{equation}
\frac{\langle z|\widehat{F}|z'\rangle}{\langle z|z'\rangle}=
\frac{\langle \tilde{z}|\widehat{F}|z'\rangle}{\langle \tilde{z}|z'\rangle}
\left(
1+\mathcal{O}(\hbar)
\right)
\end{equation}
is proved in the same way.
\end{proof}

Now we can  prove the asymptotic relation by mathematical induction, resolving the variable coupling in multi-integrations.

\begin{lemma}\label{theoremAsymp}
For two coherent states $|z\rangle$ and $|z'\rangle$, given an operator $\widehat{A}$ which is a polynomial of fundamental operators, and the Berezin symbol $C(z_1, z_2) \equiv \frac{\langle z_1 | \hat{A} | z_2 \rangle}{\langle z_1 | z_2 \rangle}\neq0$ for arbitrary coherent states $| z_1 \rangle$ and $| z_2 \rangle$, the following asymptotic relation holds for any positive integer $N$:
\begin{align}
\frac{\langle z| \left( \hat{A}^2 - C^2(z,z') \right)^{N} |z' \rangle}{\langle z | z' \rangle} \sim \mathcal{O}\left(\hbar^{\lceil \frac{N}{2} \rceil}\right), \quad N=1,2,3...
\end{align}
provided the following requirements are satisfied: \\
(1) The semi-classical fluctuation behaves as $\frac{\langle z |\left[\widehat{A}^{2}- C^2(z,z') \right] | z'\rangle}{\langle z|z'\rangle} = \hbar f(z,z')$ for all $z,z'$, where $f(z,z')$ is a smooth function. \\
(2) The Kähler potential satisfies $\Re\left[K\left(z, \bar{z}^{\prime}\right)\right] \le \frac{1}{2}\left[K(z, \bar{z})+K\left(z^{\prime}, \bar{z}^{\prime}\right)\right]$ for all $z,z'$ where the equality is taken as $z=z'$ to ensure an isolated dominant critical point, and the metric $g_{z \bar{z}}>0$ ensures non-degeneracy.
\end{lemma}

\begin{proof}
Let us define
\begin{equation}
\widehat{O}_{z,z'} \equiv \widehat{A}^{2} - C^2(z,z').
\end{equation}
By condition (1), its off-diagonal matrix element satisfies
\begin{equation}
\langle z| \widehat{O}_{z,z'} |z'\rangle = \hbar \langle z|z'\rangle f(z,z'),
\end{equation}
with $f(z,z')$ smooth. We now prove the statement by induction on $N$.

\textbf{Base case: $N=1$.}

This is exactly condition (1):
\begin{equation}
\frac{\langle z| \widehat{O}_{z,z'} |z'\rangle}{\langle z|z'\rangle}=
\hbar f(z,z')
\mathcal{O}(\hbar),
\end{equation}
which agrees with $\mathcal{O}(\hbar^{\lceil 1/2\rceil})$.

\textbf{Base case: $N=2$.}

Insert one resolution of identity:
\begin{equation}\label{Q4_new}
\begin{split}
\langle z| \widehat{O}_{z,z'}^2 |z'\rangle
&=
\int dv(z_1)\,
\langle z| \widehat{O}_{z,z'} |z_1\rangle
\langle z_1| \widehat{O}_{z,z'} |z'\rangle\\
&=
\int dv(z_1)\,
\langle z|z_1\rangle \langle z_1|z'\rangle
\left[
\frac{\langle z|\widehat{A}^2|z_1\rangle}{\langle z|z_1\rangle}
- C^2(z,z')
\right]
\left[
\frac{\langle z_1|\widehat{A}^2|z'\rangle}{\langle z_1|z'\rangle}
- C^2(z,z')
\right].
\end{split}
\end{equation}
Using condition (1) with $z'$ replaced by $z_1$, we write
\begin{equation}
\frac{\langle z|\widehat{A}^2|z_1\rangle}{\langle z|z_1\rangle}=
C^2(z,z_1)+\hbar f(z,z_1),
\qquad
\frac{\langle z_1|\widehat{A}^2|z'\rangle}{\langle z_1|z'\rangle}=
C^2(z_1,z')+\hbar f(z_1,z').
\end{equation}
Hence the amplitude in the integral takes the form
\begin{equation}
u(z_1)=
\left[
C^2(z,z_1)-C^2(z,z')+\hbar f(z,z_1)
\right]
\left[
C^2(z_1,z')-C^2(z,z')+\hbar f(z_1,z')
\right].
\end{equation}
The phase is again given by the geometric overlap action $S(z,z_1,z')$. By Lemma \ref{lemma1}, at the dominant saddle point $\tilde{z}$ one has
\begin{equation}
C(z,\tilde{z})=C(z,z')+\mathcal{O}(\hbar),
\qquad
C(\tilde{z},z')=C(z,z')+\mathcal{O}(\hbar),
\end{equation}
so the leading $\mathcal{O}(\hbar^0)$ contribution to the amplitude vanishes at the saddle. Therefore the zeroth stationary-phase contribution is absent. The first non-vanishing term can only arise when derivatives from the operator $L_s$ in Theorem \ref{theorem111} act on the vanishing amplitude. Since the first possible such term comes from $L_1$, it carries one explicit power of $\hbar$. Consequently,
\begin{equation}
\langle z| \widehat{O}_{z,z'}^2 |z'\rangle=
\hbar \langle z|z'\rangle \mathrm{Pol}_1(z,z')
+\mathcal{O}(\hbar^2),
\end{equation}
which is of order $\mathcal{O}(\hbar^{\lceil 2/2\rceil})=\mathcal{O}(\hbar)$.

\textbf{Inductive step.}

Assume that for every $m\leq N$,
\begin{equation}\label{asumpO}
\langle z| \widehat{O}_{z,z'}^{m} |z'\rangle=
\hbar^{\lceil m/2 \rceil}
\langle z|z'\rangle \mathrm{Pol}_m(z,z').
\end{equation}
We prove the statement for $N+1$.

Insert one resolution of identity and separate one power on the right:
\begin{equation}\label{ON1}
\langle z| \widehat{O}_{z,z'}^{N+1} |z'\rangle=
\int dv(z_1)\,
\langle z|z_1\rangle \langle z_1|z'\rangle
U_L^{N}(z_1) U_R^1(z_1),
\end{equation}
where
\begin{equation}
U_L^{N}(z_1):=
\frac{\langle z| \widehat{O}_{z,z'}^{N} |z_1\rangle}{\langle z|z_1\rangle},
\qquad
U_R^{1}(z_1):=
\frac{\langle z_1| \widehat{O}_{z,z'} |z'\rangle}{\langle z_1|z'\rangle}.
\end{equation}
Introduce the local drifts
\begin{equation}
V_L(z_1)\equiv C^2(z,z_1)- C^2(z,z'),
\qquad
V_R(z_1)\equiv C^2(z_1,z')- C^2(z,z').
\end{equation}
Then
\begin{equation}
\widehat{O}_{z,z'}=\widehat{O}_{z,z_1}+V_L(z_1)
\qquad\text{on the left endpoint sector,}
\end{equation}
and similarly
\begin{equation}
\widehat{O}_{z,z'}=\widehat{O}_{z_1,z'}+V_R(z_1)
\qquad\text{on the right endpoint sector.}
\end{equation}
Therefore
\begin{equation}
U_L^{N}(z_1)=
\sum_{j=0}^{N} \binom{N}{j}
\left(
\frac{\langle z| \widehat{O}_{z,z_1}^{j} |z_1\rangle}{\langle z|z_1\rangle}
\right)
V_L(z_1)^{N-j},
\end{equation}
and
\begin{equation}
U_R^{1}(z_1)=
\sum_{i=0}^{1} \binom{1}{i}
\left(
\frac{\langle z_1| \widehat{O}_{z_1,z'}^{i} |z'\rangle}{\langle z_1|z'\rangle}
\right)
V_R(z_1)^{1-i}.
\end{equation}
By the inductive hypothesis and Lemma \ref{lemma1}, the operator matrix elements contribute orders
\begin{equation}
\mathcal{O}\left(\hbar^{\lceil j/2\rceil}\right)
\qquad\text{and}\qquad
\mathcal{O}\left(\hbar^{\lceil i/2\rceil}\right),
\end{equation}
respectively.

It remains to estimate the contribution of the drift factors $V_L^{N-j}V_R^{1-i}$ under the saddle-point expansion. By Lemma \ref{lemma1}, both $V_L$ and $V_R$ vanish at the saddle point $\tilde{z}$ at leading order. Hence the product
\begin{equation}
V_L(z_1)^{N-j}V_R(z_1)^{1-i}
\end{equation}
vanishes to order
\begin{equation}
n=(N-j)+(1-i)=N+1-j-i
\end{equation}
at the saddle. A nonzero stationary-phase contribution can therefore arise only when sufficiently many derivatives from the operator $L_s$ act on this vanishing factor. Since $L_s$ contains $2s$ derivatives and is weighted by $\hbar^s$, the smallest additional semiclassical power that can compensate a vanishing order $n$ is
\begin{equation}
\hbar^{\lceil n/2\rceil}=
\hbar^{\left\lceil \frac{N+1-j-i}{2}\right\rceil}.
\end{equation}
Combining this with the orders already coming from the operator insertions, each term in the binomial expansion contributes at least
\begin{equation}
\mathcal{O}\left(
\hbar^{
\left\lceil \frac{j}{2} \right\rceil
+
\left\lceil \frac{i}{2} \right\rceil
+
\left\lceil \frac{N+1-j-i}{2} \right\rceil
}
\right).
\end{equation}
Using the elementary inequality
\begin{equation}
\left\lceil a \right\rceil + \left\lceil b \right\rceil + \left\lceil c \right\rceil
\geq
\left\lceil a+b+c \right\rceil,
\end{equation}
we obtain for every pair $(j,i)$
\begin{equation}
\left\lceil \frac{j}{2} \right\rceil
+
\left\lceil \frac{i}{2} \right\rceil
+
\left\lceil \frac{N+1-j-i}{2} \right\rceil
\geq
\left\lceil \frac{N+1}{2} \right\rceil.
\end{equation}
Therefore
\begin{equation}
\frac{\langle z| \widehat{O}_{z,z'}^{N+1} |z'\rangle}{\langle z|z'\rangle}=
\mathcal{O}\left(\hbar^{\left\lceil \frac{N+1}{2}\right\rceil}\right),
\end{equation}
which completes the induction.
\end{proof}

\subsection{The main expansion fomula}
In this subsection, in order to further expand equation (\ref{zexpand111}), we will prove the following asymptotic expansion formula of the Berezin symbol of $(\widehat{A}^2)^q$. 
\begin{theorem1}
For two coherent states $|z\rangle$ and $|z'\rangle$ defined by Eqs. (\ref{defz}), (\ref{prop1}), and (\ref{rI}), let $\widehat{A}$ be a symmetric operator which is a polynomial in the fundamental operators of the theory. Define the off-diagonal Berezin symbol
\begin{equation}
C(z,z') \equiv \frac{\langle z|\widehat{A}|z'\rangle}{\langle z|z'\rangle}.
\end{equation}
Assume that:
\noindent
(1) $C(z,z')$ is non-vanishing for all $z,z'$, so that the branch of the fractional power is well-defined; and the off-diagonal fluctuation of $\widehat{A}^2$ is semiclassically small, namely
\begin{equation}
\frac{\langle z|\widehat{A}^2|z'\rangle\langle z|z'\rangle}{\langle z|\widehat{A}|z' \rangle^2}=
1+\mathcal{O}(\hbar),
\end{equation}
equivalently,
\begin{equation}
\frac{\langle z|(\widehat{A}^2-C^2(z,z'))|z'\rangle}{\langle z|z'\rangle}=
\mathcal{O}(\hbar);
\end{equation}
\noindent
(2) the K\"ahler potential (\ref{Kpotential}) satisfies
\begin{equation}
\Re\left[K\left(z, \bar{z}^{\prime}\right)\right]
< \frac{1}{2}\left[K(z, \bar{z})+K\left(z^{\prime}, \bar{z}^{\prime}\right)\right], \qquad \forall z\neq z', \end{equation} so that the off-diagonal overlap is exponentially suppressed away from the dominant saddle, and the K\"ahler metric is strictly positive definite, \begin{equation} g_{z \bar{z}}=\frac{\partial^2 K(z, \bar{z})}{\partial z \partial \bar{z}}>0,
\end{equation}
ensuring a non-degenerate Hessian for the saddle-point action.

Then, for every positive integer $k$,
\begin{equation}\label{expan1}
\begin{split}
[(\widehat{A}^2)^q](z,z')&=\frac{\langle z |(\widehat{A}^2)^q|z' \rangle}{\langle z |z' \rangle}=C^{2q}(z,z')
\left[
1
+
\sum_{n=1}^{2k}
\binom{q}{n}
\frac{\langle z| \left( \frac{\widehat{A}^2}{C^2(z,z')} - 1 \right)^n |z' \rangle}{\langle z | z' \rangle}
\right]
+
\mathcal{O}(\hbar^{k+1}).
\end{split}
\end{equation}
\end{theorem1}

The proof of this result can be outlined as follows:
First, the operator is formally projected onto the space of eigenvectors using the spectral theorem. To bypass functional convergence issues for unbounded operators, the binomial is truncated using the Taylor series with Lagrange integral remainders. Second, this remainder is factorized at the operator-algebraic level. Finally, the remainder matrix element is evaluated by leveraging our previously established Lemma \ref{theoremAsymp} to bound the remainder at $\mathcal{O}(\hbar^{k+1})$.

\begin{proof}
For notational convenience, we begin with the pure-point form of the functional calculus. This introduces no restriction on the argument, because the estimate of the remainder will subsequently be reformulated entirely at the operator level and will therefore not depend on whether the spectrum is discrete or continuous.

Suppose first that $\widehat{A}^2$ has a discrete spectral resolution. Then
\begin{equation}
(\widehat{A}^2)^q=
\sum_{\lambda \in \sigma(A)} \lambda^{2q} E(\lambda),
\end{equation}
where $E(\lambda)$ denotes the spectral projector. Accordingly,
\begin{equation}
\frac{\langle z|(\widehat{A}^2)^q|z'\rangle}{\langle z|z'\rangle}=
\sum_{\lambda \in \sigma(A)}
\lambda^{2q}
\mu_\lambda(z,z'),
\end{equation}
with
\begin{equation}
\mu_\lambda(z,z')
:=
\frac{\langle z|E(\lambda)|z'\rangle}{\langle z|z'\rangle}.
\end{equation}
Similarly, the Berezin symbol of $\widehat{A}$ is
\begin{equation}
C(z,z')=
\frac{\langle z|\widehat{A}|z'\rangle}{\langle z|z'\rangle}=
\sum_{\lambda \in \sigma(A)} \lambda \,\mu_\lambda(z,z'),
\end{equation}
up to the chosen branch compatible with the operator under consideration. For the present argument, however, the only essential point is that the expansion is organized around the off-diagonal symbol $C(z,z')$.

Introduce
\begin{equation}
x_\lambda:=\frac{\lambda^2}{C^2(z,z')}-1.
\end{equation}
Then
\begin{equation}
[(\widehat{A}^2)^q](z,z')=
C^{2q}(z,z')
\sum_{\lambda}
(1+x_\lambda)^q \mu_\lambda(z,z').
\end{equation}
Instead of using the full binomial series directly, we truncate it algebraically at finite order and keep the Taylor remainder:
\begin{equation}
(1+x_\lambda)^q=
\sum_{n=0}^{N}\binom{q}{n}x_\lambda^n
+
R_N(x_\lambda).
\end{equation}
Hence
\begin{equation}
[(\widehat{A}^2)^q](z,z')=
C^{2q}(z,z')
\sum_{n=0}^{N}\binom{q}{n}\sum_\lambda x_\lambda^n \mu_\lambda(z,z')
+
\mathcal{E}_N,
\end{equation}
where
\begin{equation}\label{error}
\mathcal{E}_N
:=
C^{2q}(z,z')
\sum_\lambda R_N(x_\lambda)\mu_\lambda(z,z').
\end{equation}
The finite sum gives the truncated expansion in Eq. (\ref{expan1}); it therefore remains to estimate $\mathcal{E}_N$.

By Taylor's theorem with integral remainder,
\begin{equation}\label{intRN}
R_N(x)=
x^{N+1}\Phi_N(x),
\end{equation}
where
\begin{equation}\label{intphi}
\Phi_N(x)=
(-1)^{N+1}\frac{C_{N,q}}{N!}
\int_0^1 dt\,(1-t)^N(1+tx)^{q-N-1}.
\end{equation}
For the argument below, the only property of $\Phi_N$ that is needed is that it is analytic, and hence bounded, in a neighborhood of the origin containing the relevant semiclassical values of $x$. It is also useful to note that the integral remainder remains well behaved at the physical boundary of the spectrum. In the extreme case $x\to -1$, corresponding to $\lambda\to 0$, the factor $(1+tx)^{q-N-1}$ reduces to $(1-t)^{q-N-1}$, so that the integrand in Eq.~(\ref{intphi}) behaves as $(1-t)^{q-1}$. Since $q>0$ in the applications of interest, the integral $\int_0^1 dt\,(1-t)^{q-1}$ remains absolutely convergent. Therefore $\Phi_N(\widehat X)$ does not develop singular behavior at the lower edge of the physical spectrum.

We now rewrite the remainder in a way that no longer depends on the explicit spectral type. Define
\begin{equation}
\widehat{X}
:=
\frac{\widehat{A}^2}{C^2(z,z')}-1=
\frac{\widehat{A}^2-C^2(z,z')}{C^2(z,z')}=
\frac{\widehat{O}_{z,z'}}{C^2(z,z')}.
\end{equation}
Then the factorization (\ref{intRN}) lifts to the operator identity
\begin{equation}
R_N(\widehat{X})=
\widehat{X}^{N+1}\Phi_N(\widehat{X}),
\end{equation}
and therefore
\begin{equation}\label{exact_op_E}
\mathcal{E}_N=
\frac{\langle z|R_N(\widehat{X})|z'\rangle}{\langle z|z'\rangle}=
\frac{\langle z|\widehat{X}^{N+1}\Phi_N(\widehat{X})|z'\rangle}{\langle z|z'\rangle}.
\end{equation}
Insert one resolution of identity:
\begin{equation}\label{E_integral}
\begin{split}
\mathcal{E}_N
&=
\int dv(z_1)\,
\frac{\langle z|\widehat{X}^{N+1}|z_1\rangle
\langle z_1|\Phi_N(\widehat{X})|z'\rangle}{\langle z|z'\rangle}\\
&=
\frac{1}{\langle z|z'\rangle}
\int dv(z_1)\,
e^{\frac{1}{\hbar}S(z,z_1,z')}
u_1(z_1)u_2(z_1),
\end{split}
\end{equation}
where
\begin{equation}
u_1(z_1):=
\frac{\langle z|\widehat{X}^{N+1}|z_1\rangle}{\langle z|z_1\rangle},
\qquad
u_2(z_1):=
\frac{\langle z_1|\Phi_N(\widehat{X})|z'\rangle}{\langle z_1|z'\rangle}.
\end{equation}
To prove that the operator matrix element $u_2(z_1)$ is a smooth macroscopic amplitude without relying on potentially divergent power series, we directly utilize the integral representation of the Taylor remainder. Because the physical spectrum requires $x \ge -1$, the term $(1+tx) \ge 1-t \ge 0$ is positive and bounded away from any singularities for all $t \in [0,1)$. Consequently, $\Phi(x)$ is a globally smooth and well-behaved function. Commuting the semiclassical limit with the compact integration over $t$, $u_2(z_1)$ strictly evaluates to its classical principal symbol:
\begin{equation}
u_2(z_1) = \Phi\left( \frac{\langle z_1 | \widehat{X} | z' \rangle}{\langle z_1 | z' \rangle} \right) + \mathcal{O}(\hbar) = \mathcal{O}(1).
\end{equation}
Thus the semiclassical order of $\mathcal{E}_N$ is controlled entirely by $u_1(z_1)$ together with the stationary-phase expansion of the geometric overlap.

The saddle-point expansion of the integral (\ref{E_integral}) is governed by the
isolated critical point $\tilde{z}$ satisfying $\nabla_{z_1}S=0$.
Introduce the local drift
\begin{equation}
V(z_1)\equiv C^2(z,z_1)-C^2(z,z'),
\end{equation}
and observe that $\widehat{O}_{z,z'}=\widehat{O}_{z,z_1}+V(z_1)$.
Since $\widehat{X}=\widehat{O}_{z,z'}/C^2(z,z')$, the amplitude $u_1$ admits the binomial expansion
\begin{equation}\label{u1_expand}
\begin{split}
u_1(z_1)
&=\frac{\langle z|\widehat{X}^{N+1}|z_1\rangle}{\langle z|z_1\rangle}
=\frac{1}{C^{2(N+1)}(z,z')}
\frac{\langle z|\left(\widehat{O}_{z,z_1}+V(z_1)\right)^{N+1}|z_1\rangle}
     {\langle z|z_1\rangle}\\
&=\frac{1}{C^{2(N+1)}(z,z')}
  \sum_{j=0}^{N+1}\binom{N+1}{j}
  \left(
    \frac{\langle z|\widehat{O}_{z,z_1}^{j}|z_1\rangle}{\langle z|z_1\rangle}
  \right)
  V(z_1)^{N+1-j}.
\end{split}
\end{equation}
The prefactor $C^{-2(N+1)}(z,z')$ is a $c$-number independent of $z_1$ and $\hbar$
and therefore does not affect the semiclassical order.
By Theorem~\ref{theoremAsymp}, the operator matrix elements contribute
$\mathcal{O}\!\left(\hbar^{\lceil j/2\rceil}\right)$
for arbitrary endpoints, independently of the saddle-point evaluation.
By Lemma~\ref{lemma1}, the drift $V(z_1)$ vanishes at leading order at
$z_1=\tilde{z}$.
To obtain a non-vanishing contribution, the factor $V(z_1)^{N+1-j}$,
which vanishes to order $n=N+1-j$ at the saddle, must be differentiated
by the H\"ormander operator $L_s$ in Theorem~\ref{theorem111}.
Since $L_s$ carries $2s$ derivatives and is weighted by $\hbar^s$,
the minimal semiclassical power required to compensate $n$ vanishing orders
is $\hbar^{\lceil n/2\rceil}$.
Combining the two contributions, each $j$-th term is bounded by
\begin{equation}
\mathcal{O}\!\left(
  \hbar^{\left\lceil\frac{j}{2}\right\rceil
        +\left\lceil\frac{N+1-j}{2}\right\rceil}
\right)
\;\ge\;
\mathcal{O}\!\left(
  \hbar^{\left\lceil\frac{N+1}{2}\right\rceil}
\right),
\end{equation}
where the inequality follows from the sub-additivity
$\lceil a\rceil+\lceil b\rceil\ge\lceil a+b\rceil$.
Because $u_2(z_1)=\mathcal{O}(1)$, derivatives of the composite amplitude
$u_1 u_2$ can only generate higher-order corrections.
Setting the truncation order to $N=2k$, we obtain
\begin{equation}
\mathcal{E}_{2k}
=\mathcal{O}\!\left(
  \hbar^{\left\lceil\frac{2k+1}{2}\right\rceil}
\right)
=\mathcal{O}(\hbar^{k+1}),
\end{equation}
which completes the proof.
\end{proof}

\textbf{Remark.}\;
Although our main result is formulated for $(\widehat{A}^2)^q$,
the underlying method extends to a broad class of analytic functions
$f(\widehat{A})$.
One expands $f$ around the off-diagonal symbol
$C(z,z')=\langle z|\widehat{A}|z'\rangle/\langle z|z'\rangle$
via a Taylor series with integral remainder and factorizes
the remainder as $R_N(\widehat{X})=\widehat{X}^{N+1}\Phi_N(\widehat{X})$
at the operator level.
The function $\Phi_N$ depends on $f$ only through its values and
derivatives evaluated at $x=0$ in the integral representation
(\ref{intRN}); in particular, the change from $(1+x)^q$ to a more
general $f$ affects only the prefactor outside the integration variable
$t$ in Eq.~(\ref{intphi}) and does not modify the analytic structure
that guarantees $\Phi_N(\widehat{X})=\mathcal{O}(1)$.
The saddle-point evaluation and the zero-root compensation mechanism of lemma~\ref{theoremAsymp} are therefore expected to apply in the same way, under analogous analyticity, non-vanishing, and semiclassical regularity assumptions.

\section{Application in LQG}\label{sectionIV}

In the deparametrized framework of canonical LQG,
matter fields (e.g.\ scalar or Gaussian dust) serve as
physical clocks and rods, and the resulting coherent-state path integral
requires the perturbative evaluation of matrix elements at each
intermediate time step, in direct analogy with Eq.~(\ref{cpath}).
Because the physical Hamiltonian is built from holonomy and flux
operators, we must show that polynomials of these operators
satisfy the hypotheses of the expansion~(\ref{expan1}).

The section is organized as follows.
Subsection~\ref{sec:LQGcs} reviews the coherent-state representation
used in LQG.
Subsection~\ref{sec:LQGexpand} proves that flux-operator functions
can be expanded via our main result, re-expresses the
holonomy action in terms of the flux operator. Subsection~\ref{sec:LQGN} presents a
comparison of the expansion formula against the numerical quantum
matrix elements.

\subsection{LQG coherent states}\label{sec:LQGcs}

The Thiemann complexifier coherent states
\cite{Thiemann:2000bw,Thiemann:2000ca,Thiemann:2000bx,Thiemann:2000by,Sahlmann:2001nv,Thiemann:2002vj}
constitute the principal tool for the semiclassical analysis of LQG.
They satisfy the over-completeness relation
(the third property in Eq.~(\ref{prop1}))
and possess two additional features:
(i) the coherent-state label parametrizes the LQG phase space, and
(ii) the overlap function is sharply peaked in phase space.

The single-edge coherent state $|g(e)\rangle$ associated with an edge
$e\in E(\Gamma)$ is defined as a function of the $SU(2)$ element $h(e)$
\cite{Thiemann:2000ca}:
\begin{equation}\label{coherent1}
\psi^{t}_{g(e)}(h(e))
\equiv\langle h(e)|g(e)\rangle
=\sum_{j_e\in\mathbb{Z}_+/2\cup\{0\}}
  d_{j_e}\,e^{-t\lambda_{j_e}/2}\,
  \chi_{j_e}\!\left(g(e)h^{-1}(e)\right),
\end{equation}
where $\chi_j$ is the $SU(2)$ character in spin-$j$,
$d_j:=2j+1$, $\lambda_j:=j(j+1)$, and
$t=l_p^2/a^2\propto\hbar$ is the semiclassical parameter that vanishes
in the classical limit.
The label $g(e)\in SL(2,\mathbb{C})$ is the complexified holonomy
\begin{equation}\label{groupg}
g(e)=e^{-ip^a(e)\tau_a/2}\,e^{z^a(e)\tau_a/2},
\qquad p^a(e),\,z^a(e)\in\mathbb{R}^3,
\end{equation}
with $\tau_a=-i\sigma_a$ and $\sigma_a$ ($a=1,2,3$) the Pauli matrices.

The overlap function reads \cite{Thiemann:2000ca}
\begin{equation}\label{LQGoverlap}
\langle\psi^t_g|\psi^t_{g'}\rangle
=\frac{\sqrt{\pi}\,e^{t/4}}{4\sinh(z_0)\,T^3}
 \sum_{n\in\mathbb{Z}}(z_0-2\pi i n)\,
 e^{(z_0-2\pi in)^2/t}
\;\approx\;
\frac{\sqrt{\pi}\,e^{t/4}}{4\sinh(z_0)\,T^3}\,
z_0\,e^{z_0^2/t},
\end{equation}
where $T\equiv\sqrt{t}/2$ and
$z_0 = m(g,g')$ with $m(g,g'):= \cosh^{-1}(\frac{1}{2}\operatorname{tr}(g'g^\dagger))$.
The approximation retains only the $n=0$ term; all others are
exponentially suppressed as $\sim e^{-4n^2\pi^2/t}$ for small $t$.
This approximation is adopted throughout.

Using the Poisson summation formula, the single-edge norm is
\cite{Thiemann:2000ca}
\begin{equation}\label{Normalization}
\|\psi^t_g\|^2
=\psi^{2t}_{H^2}(1)
\;\approx\;
\frac{\sqrt{\pi}\,e^{t/4}}{4\sinh(y_g)\,T^3}\,
y_g\,e^{y_g^2/t},
\end{equation}
where $H=e^{-ip^a(e)\tau^a/2}$ is the boost part of $g$ and
$\cosh(y_g):=\tfrac{1}{2}\operatorname{tr}(H^2)$.
The normalized overlap therefore takes the form
\begin{equation}\label{LQGoverlapNorm}
\langle\tilde{\psi}^t_g|\tilde{\psi}^t_{g'}\rangle
\;\approx\;
\frac{\sqrt{\sinh(y_g)\sinh(y_{g'})}}{\sinh(z_0)}\,
\frac{z_0}{\sqrt{y_g\,y_{g'}}}\;
e^{\left(z_0^2-\frac{1}{2}y_g^2-\frac{1}{2}y_{g'}^2\right)/t},
\end{equation}
where $|\tilde{\psi}^t_g\rangle$ denotes the normalized coherent state.
The prefactor before the exponential originates from the
integration measure, and the exponential term
$e^{(z_0^2-y_g^2/2-y_{g'}^2/2)/t}$ recovers exactly the third
property of Eq.~(\ref{prop1}).

For a graph $\Gamma$ with $L$ edges $e_1,\dots,e_L$, the multi-edge
coherent state is the product
\begin{equation}\label{stateL}
\psi^t_{\Gamma,\{g\}}(h(e_1),\dots,h(e_L))
=\prod_{l=1}^{L}
\left(
  \sum_{j_l\in\mathbb{Z}_+/2\cup\{0\}}
  d_{j_l}\,e^{-t\lambda_{j_l}/2}\,
  \chi_{j_l}\!\left(g(e_l)h^{-1}(e_l)\right)
\right),
\end{equation}
where $\{g\}$ denotes the collection of complexified holonomies for
every edge in $\Gamma$.
The gauge-invariant coherent state is obtained by group averaging
\cite{PhysRevD.92.104023}:
\begin{equation}\label{groupA}
\Psi^t_{\Gamma,\{g\}}(h(e_1),\dots,h(e_L))
=\int_{SU(2)^N}
  d\mu_H(\tilde{g}_1)\cdots d\mu_H(\tilde{g}_N)\;
  \psi^t_{\{g\}}\!\left(
    \tilde{g}_{t(1)}h(e_1)\tilde{g}^{-1}_{s(1)},\dots,
    \tilde{g}_{t(L)}h(e_L)\tilde{g}^{-1}_{s(L)}
  \right),
\end{equation}
where $s(l)$ and $t(l)$ are the source and target vertices of edge~$l$.

\subsection{Application of the expansion formula in LQG}
\label{sec:LQGexpand}

The goal is to construct a perturbative expansion for the matrix
elements of the physical Hamiltonian operator.
The general matrix element to be evaluated is \cite{Giesel:2006um}
\begin{equation}\label{generalM}
\langle\psi|\,
  f_1(\widehat{h})\,F_1(\widehat{V}_{v_1})\,
  f_2(\widehat{h})\,F_2(\widehat{V}_{v_2})\,\cdots\,
  F_N(\widehat{V}_{v_N})\,f_{N+1}(\widehat{h})
\,|\psi'\rangle,
\end{equation}
where the $f_i$ are polynomials of the holonomy operator and
$F_I(\widehat{V}_{v_I})\propto(\widehat{Q}_{v_I}^2)^{q_I}$
with $0<q_I=m_I/n_I\le 1/4$ a rational number ($m_I$, $n_I$ coprime).
The self-adjoint operator $\widehat{Q}_v$ is
\begin{equation}\label{Qdef}
\widehat{Q}_v
=\sum_{\substack{e_I,e_J,e_K\in E(\gamma)\\
                  e_I\cap e_J\cap e_K=v}}
  \epsilon(e_I(v),e_J(v),e_K(v))\,
  \epsilon_{ijk}\,
  \widehat{p}^{\,i}_{e_I(v)}\,
  \widehat{p}^{\,j}_{e_J(v)}\,
  \widehat{p}^{\,k}_{e_K(v)},
\end{equation}
where the sum runs over all ordered triples of edges incident at $v$
($i,j,k=1,2,3$), $\epsilon(e_I(v),e_J(v),e_K(v))$ is the orientation
sign of the three tangent vectors, and the flux components
$\widehat{p}^{\,i}_{e_I}=it\,X^i_{e_I}/2$ are defined through the
right-invariant vector fields
$X^i_{e_I}=\operatorname{tr}\!\left((\tau^i h_{e_I})^T
\partial/\partial h_{e_I}\right)$ on $SU(2)$.

Upon inserting the resolution of identity between each pair of
operators, Eq.~(\ref{generalM}) becomes
\begin{equation}\label{generalM2}
\begin{split}
\int\cdots\int d\mu\;\bigl[
  &\langle\psi|f_1(\widehat{h})|\psi_1\rangle\,
   \langle\psi_1|F_1(\widehat{V}_{v_1})|\psi_2\rangle\,
   \langle\psi_2|f_2(\widehat{h})|\psi_3\rangle\,
   \langle\psi_3|F_2(\widehat{V}_{v_2})|\psi_4\rangle\;\cdots\\
  &\qquad\cdots\;
   \langle\psi_{2N-1}|F_N(\widehat{V}_{v_N})|\psi_{2N}\rangle\,
   \langle\psi_{2N}|f_{N+1}(\widehat{h})|\psi'\rangle
\bigr],
\end{split}
\end{equation}
where $d\mu$ is the collective measure of the inserted resolutions.

For each volume-type factor $F[\widehat{V}_v]$, our main result
yields the expansion
\begin{equation}\label{VexpandNew}
\frac{\langle\psi|F[\widehat{V}_v]|\psi'\rangle}
     {\langle\psi|\psi'\rangle}
\;\sim\;
\left(
  \frac{\langle\psi|\widehat{Q}_v|\psi'\rangle}
       {\langle\psi|\psi'\rangle}
\right)^{\!2q}
\left[
  1+\sum_{N=1}^{2k}\binom{q}{N}
  \frac{\langle\psi|
    \left(
      \frac{\langle\psi|\psi'\rangle^2\,\widehat{Q}_v^2}
           {\langle\psi|\widehat{Q}_v|\psi'\rangle^2}-1
    \right)^{\!N}
  |\psi'\rangle}
  {\langle\psi|\psi'\rangle}
\right]
+\mathcal{O}(t^{k+1}).
\end{equation}
This replaces the earlier expansion proposed in \cite{Giesel:2006um}:
\begin{equation}\label{Vexpand}
\frac{\langle\psi|F[\widehat{V}_v]|\psi'\rangle}
     {\langle\psi|\psi'\rangle}
\;\approx\;
\langle\tilde{\psi}|\widehat{Q}_v|\tilde{\psi}\rangle^{2q}
\left[
  1+\sum_{N=1}^{2k+1}
  \binom{q}{N}
  \frac{\langle\psi|
    \left(
      \frac{\widehat{Q}_v^2}
           {\langle\tilde{\psi}|\widehat{Q}_v|\tilde{\psi}\rangle^2}-1
    \right)^{\!N}
  |\psi'\rangle}
  {\langle\psi|\psi'\rangle}
\right]+\mathcal{O}(t^{k+1}),
\end{equation}
where $\tilde{\psi}$ denotes the normalized LQG coherent state.
The key deficiency of Eq.~(\ref{Vexpand})
is that it replaces the full off-diagonal matrix elements with
diagonal expectation values of a single coherent state
$|\psi^t_g\rangle$.
While this may be adequate when $|\psi^t_g\rangle$ and
$|\psi^t_{g'}\rangle$ have strong overlap, it generically introduces
leading-order errors for distinct states.
The new expansion~(\ref{VexpandNew}) resolves this by systematically
using genuine off-diagonal matrix elements, ensuring consistency in
the semiclassical limit.

Recall that the non-vanishing assumption $C(z,z') = \frac{\langle\psi|\widehat{Q}_v|\psi'\rangle}{\langle\psi|\psi'\rangle} \neq 0$ is essential for organizing the expansion around the off-diagonal Berezin symbol and for fixing the branch of the fractional power. In the LQG applications considered here, this requirement is naturally satisfied in non-degenerate semiclassical sectors carrying non-vanishing geometric flux. However, the kinematical state space also contains degenerate configurations, including zero-volume or vanishing-flux geometries, for which the off-diagonal symbol may vanish or become parametrically small. In such regions the expansion variable $\widehat X=(\widehat A^2/C^2)-1$ becomes singular, and the present asymptotic expansion ceases to be uniformly valid. Extending the formalism to transitions probing such degenerate sectors would require an additional regularization or a different choice of expansion center.

We now verify that the hypotheses of our main result are
satisfied and then present a numerical comparison against
quantum matrix elements computed with the algorithm of
\cite{Liliu:202603,LQG-Volume}.

\subsubsection*{Verification of condition~(1)}

Consider first the action of a single flux component $\widehat{p}^j_e$
on the single-edge coherent state:
\begin{equation}\label{paction}
\widehat{p}^j_e\,\psi^t_g(h)
:=\frac{it}{2}\left(\frac{d}{ds}\right)_{\!s=0}
  \psi^t_g(e^{s\tau_j}h).
\end{equation}
The matrix element evaluates to
\begin{equation}
\langle\psi^t_g,\widehat{p}^j_e\,\psi^t_{g'}\rangle
=\frac{it}{2}\left(\frac{d}{ds}\right)_{\!s=0}
 \psi^{2t}_{e^{-s\tau_j}g'g^\dagger}(1).
\end{equation}
Define the complex variable $z$ through
\begin{equation}\label{zexpand}
\cosh(z)
=\tfrac{1}{2}\operatorname{tr}\!\left(e^{-s\tau_j}g'g^\dagger\right)
=\tfrac{1}{2}\left[
  \operatorname{tr}(g'g^\dagger)
  -s\operatorname{tr}(\tau_j g'g^\dagger)
\right]+\mathcal{O}(s^2).
\end{equation}
Applying the Poisson re-summation formula and retaining only the $n=0$
term gives
\begin{equation}\label{zmat}
\langle\psi^t_g,\widehat{p}^j_e\,\psi^t_{g'}\rangle
=\frac{it}{2}\left(\frac{d}{ds}\right)_{\!s=0}
 \frac{\sqrt{\pi}\,e^{t/4}}{4\sinh(z)\,T^3}\,
 z\,e^{z^2/t}.
\end{equation}
Setting $z=z_0+\delta$ with
$\cosh(z_0):=\tfrac{1}{2}\operatorname{tr}(g'g^\dagger)$ and comparing
the expansion
\begin{equation}
\cosh(z)
=\cosh(z_0)+\delta\sinh(z_0)+\mathcal{O}(\delta^2)
\end{equation}
with Eq.~(\ref{zexpand}) yields
\begin{equation}
\delta=-s\,
\frac{\operatorname{tr}(\tau_j g'g^\dagger)}{2\sinh(z_0)}.
\end{equation}
The matrix element is therefore \cite{Thiemann:2000bx}
\begin{equation}
\langle\psi^t_g,\widehat{p}^j_e\,\psi^t_{g'}\rangle
=-\frac{it}{4}\,
\frac{\operatorname{tr}(\tau_j g'g^\dagger)}{\sinh(z_0)}\,
\frac{\sqrt{\pi}\,e^{t/4}}{4T^3}\,e^{z_0^2/t}
\left[
  \frac{2z_0^2}{t\sinh(z_0)}
  -\frac{z_0\cosh(z_0)}{\sinh^2(z_0)}
  +\frac{1}{\sinh(z_0)}
\right].
\end{equation}
Dividing by the overlap~(\ref{LQGoverlap}),
\begin{equation}\label{Q1Mat1}
\frac{\langle\psi^t_g|\widehat{p}^j_e|\psi^t_{g'}\rangle}
     {\langle\psi^t_g|\psi^t_{g'}\rangle}
=-\frac{i}{2}\,
\frac{z_0\operatorname{tr}(\tau_j g'g^\dagger)}{\sinh(z_0)}
+\mathcal{O}(t).
\end{equation}

For two consecutive flux operators on the same edge,
\begin{equation}
\widehat{p}^i_e\widehat{p}^j_e\,\psi^t_g(h)
=-\frac{t^2}{4}
\left(\frac{d}{ds_1}\right)_{\!s_1=0}
\left(\frac{d}{ds_2}\right)_{\!s_2=0}
\psi^t_{e^{-s_2\tau_j}e^{-s_1\tau_i}g}(h),
\end{equation}
so that
\begin{equation}
\langle\psi^t_g|\widehat{p}^i_e\widehat{p}^j_e|\psi^t_{g'}\rangle
=-\frac{t^2}{4}
\left(\frac{d}{ds_1}\right)_{\!s_1=0}
\left(\frac{d}{ds_2}\right)_{\!s_2=0}
\frac{\sqrt{\pi}\,e^{t/4}}{4\sinh(z)\,T^3}\,
z\,e^{z^2/t},
\end{equation}
where
$\cosh z:=\tfrac{1}{2}\operatorname{tr}(e^{-s_2\tau_j}e^{-s_1\tau_i}
g'g^\dagger)$
and only the $n=0$ term is retained.
Defining
\begin{equation}
\cosh z_1:=\tfrac{1}{2}\operatorname{tr}(e^{-s_1\tau_i}g'g^\dagger),
\qquad
\cosh z_0:=\tfrac{1}{2}\operatorname{tr}(g'g^\dagger),
\end{equation}
the chain rule gives
\begin{equation}
\begin{split}
\left(\frac{d}{ds_2}\right)_{\!s_2=0}
&=-\frac{\operatorname{tr}(\tau_j e^{-s_1\tau_i}g'g^\dagger)}
        {2\sinh(z_1)}
  \left(\frac{d}{dz}\right)_{\!z=z_1},\\
\left(\frac{d}{ds_1}\right)_{\!s_1=0}
&=-\frac{\operatorname{tr}(\tau_i g'g^\dagger)}
        {2\sinh(z_0)}
  \left(\frac{d}{dz}\right)_{\!z=z_0}.
\end{split}
\end{equation}
Carrying out the differentiation,
\begin{equation}
\begin{split}
\langle\psi^t_g|\widehat{p}^i_e\widehat{p}^j_e|\psi^t_{g'}\rangle
={}&
-\frac{t^2}{4}
\left(\frac{d}{ds_1}\right)_{\!s_1=0}
\!\left[
  -\frac{\operatorname{tr}(\tau_j e^{-s_1\tau_i}g'g^\dagger)}
        {2\sinh(z_1)}
  \left(\frac{d}{dz}\right)_{\!z=z_1}
  \frac{\sqrt{\pi}\,e^{t/4}}{4\sinh(z)\,T^3}\,
  z\,e^{z^2/t}
\right]\\
={}&
-\frac{t^2}{4}
\!\left[
  -\frac{\left(\frac{d}{ds_1}\right)_{\!s_1=0}
         \operatorname{tr}(\tau_j e^{-s_1\tau_i}g'g^\dagger)}
        {2\sinh(z_0)}
  \left(\frac{d}{dz}\right)_{\!z=z_0}
  \frac{\sqrt{\pi}\,e^{t/4}}{4\sinh(z)\,T^3}\,
  z\,e^{z^2/t}
\right]\\
&-\frac{t^2}{4}
\!\left[
  \frac{\operatorname{tr}(\tau_j g'g^\dagger)\,
        \operatorname{tr}(\tau_i g'g^\dagger)}
       {4\sinh^2(z_0)}
  \left(\frac{d}{dz_1}\right)_{\!z_1=z_0}
  \!\!\left(
    \frac{1}{\sinh(z_1)}
    \left(\frac{d}{dz}\right)_{\!z=z_1}
    \frac{\sqrt{\pi}\,e^{t/4}}{4\sinh(z)\,T^3}\,
    z\,e^{z^2/t}
  \right)
\right].
\end{split}
\end{equation}
Dividing by the overlap,
\begin{equation}\label{Q2Mat1}
\frac{\langle\psi^t_g|\widehat{p}^i_e\widehat{p}^j_e|\psi^t_{g'}\rangle}
     {\langle\psi^t_g|\psi^t_{g'}\rangle}
=-\frac{1}{4}\,
\frac{z_0^2\operatorname{tr}(\tau_i g'g^\dagger)\,
      \operatorname{tr}(\tau_j g'g^\dagger)}
     {\sinh^2(z_0)}
+\mathcal{O}(t).
\end{equation}
Combining Eqs.~(\ref{Q1Mat1}) and (\ref{Q2Mat1}), one obtains the
single-edge fluctuation identity
\begin{equation}\label{respe}
\frac{\langle\psi^t_g|\psi^t_{g'}\rangle\,
      \langle\psi^t_g|\widehat{p}^i_e\widehat{p}^j_e|\psi^t_{g'}\rangle}
     {\langle\psi^t_g|\widehat{p}^i_e|\psi^t_{g'}\rangle\,
      \langle\psi^t_g|\widehat{p}^j_e|\psi^t_{g'}\rangle}
-1
=\mathcal{O}(t).
\end{equation}

We now lift this result to the multi-edge operator $\widehat{Q}_v$ and
show that
\begin{equation}\label{resQe}
\frac{\langle\psi^t_{\Gamma,\{g\}}|\psi^t_{\Gamma,\{g'\}}\rangle\;
      \langle\psi^t_{\Gamma,\{g\}}|\widehat{Q}_v^2|
      \psi^t_{\Gamma,\{g'\}}\rangle}
     {\langle\psi^t_{\Gamma,\{g\}}|\widehat{Q}_v|
      \psi^t_{\Gamma,\{g'\}}\rangle^2}
-1
=\mathcal{O}(t),
\end{equation}
which is condition~(1) for $\widehat{Q}_v$.
The argument relies on the following algebraic observation.
Let $(\alpha,\beta,\dots)$ and $(a,b,\dots)$ be two sequences of
complex numbers in one-to-one correspondence satisfying
\begin{equation}\label{sum_alpha}
\frac{\alpha}{a}
=\frac{\beta}{b}
=\cdots
=\kappa+\mathcal{O}(t)
\end{equation}
for a fixed complex constant $\kappa$.
Then
\begin{equation}\label{sum_alphak}
\frac{\alpha+\beta+\cdots}{a+b+\cdots}=\frac{(\kappa+\mathcal{O}(t))(a+b+\cdots)}{a+b+\cdots}
=\kappa+\mathcal{O}(t).
\end{equation}

Starting from the explicit definition of $\widehat{Q}_v$,
Eq.~(\ref{resQe}) can be rewritten as
\begin{equation}\label{multiedgeQ}
\begin{split}
&\frac{\langle\psi^t_{\Gamma,\{g\}}|\psi^t_{\Gamma,\{g'\}}\rangle\;
       \langle\psi^t_{\Gamma,\{g\}}|\widehat{Q}_v^2|
       \psi^t_{\Gamma,\{g'\}}\rangle}
      {\langle\psi^t_{\Gamma,\{g\}}|\widehat{Q}_v|
       \psi^t_{\Gamma,\{g'\}}\rangle^2}\\
={}&
\frac{\displaystyle\sum
  \epsilon(e_{I_1},e_{J_1},e_{K_1})\,
  \epsilon(e_{I_2},e_{J_2},e_{K_2})\,
  \epsilon_{i_1 j_1 k_1}\,\epsilon_{i_2 j_2 k_2}\;
  \langle\psi^t_{\Gamma,\{g\}}|
    \widehat{p}^{i_1}_{e_{I_1}}
    \widehat{p}^{i_2}_{e_{I_2}}
    \widehat{p}^{j_1}_{e_{J_1}}
    \widehat{p}^{j_2}_{e_{J_2}}
    \widehat{p}^{k_1}_{e_{K_1}}
    \widehat{p}^{k_2}_{e_{K_2}}
  |\psi^t_{\Gamma,\{g'\}}\rangle\;
  \langle\psi^t_{\Gamma,\{g\}}|\psi^t_{\Gamma,\{g'\}}\rangle}
{\displaystyle\sum
  \epsilon(e_{I_1},e_{J_1},e_{K_1})\,
  \epsilon(e_{I_2},e_{J_2},e_{K_2})\,
  \epsilon_{i_1 j_1 k_1}\,\epsilon_{i_2 j_2 k_2}\;
  \langle\psi^t_{\Gamma,\{g\}}|
    \widehat{p}^{i_1}_{e_{I_1}}
    \widehat{p}^{j_1}_{e_{J_1}}
    \widehat{p}^{k_1}_{e_{K_1}}
  |\psi^t_{\Gamma,\{g'\}}\rangle\;
  \langle\psi^t_{\Gamma,\{g\}}|
    \widehat{p}^{i_2}_{e_{I_2}}
    \widehat{p}^{j_2}_{e_{J_2}}
    \widehat{p}^{k_2}_{e_{K_2}}
  |\psi^t_{\Gamma,\{g'\}}\rangle},
\end{split}
\end{equation}
where in the second step the denominator has been expanded over the
same double index set.
By the algebraic identity~(\ref{sum_alphak}), it suffices to show that, for a fixed set of indices:
\begin{equation}
    (e_{I_1},e_{J_1},e_{K_1},e_{I_2},e_{J_2},e_{K_2},i_1,j_1,k_1,i_2,j_2,k_2),
\end{equation}
each summand satisfies
\begin{equation}\label{psum123}
\frac{\langle\psi^t_{\Gamma,\{g\}}|
  \widehat{p}^{i_1}_{e_{I_1}}
  \widehat{p}^{i_2}_{e_{I_2}}
  \widehat{p}^{j_1}_{e_{J_1}}
  \widehat{p}^{j_2}_{e_{J_2}}
  \widehat{p}^{k_1}_{e_{K_1}}
  \widehat{p}^{k_2}_{e_{K_2}}
|\psi^t_{\Gamma,\{g'\}}\rangle\;
\langle\psi^t_{\Gamma,\{g\}}|\psi^t_{\Gamma,\{g'\}}\rangle}
{\langle\psi^t_{\Gamma,\{g\}}|
  \widehat{p}^{i_1}_{e_{I_1}}
  \widehat{p}^{j_1}_{e_{J_1}}
  \widehat{p}^{k_1}_{e_{K_1}}
|\psi^t_{\Gamma,\{g'\}}\rangle\;
\langle\psi^t_{\Gamma,\{g\}}|
  \widehat{p}^{i_2}_{e_{I_2}}
  \widehat{p}^{j_2}_{e_{J_2}}
  \widehat{p}^{k_2}_{e_{K_2}}
|\psi^t_{\Gamma,\{g'\}}\rangle}
=1+\mathcal{O}(t).
\end{equation}
Two distinct cases arise here.

\textbf{Case~1: non-overlapping edge sets.}\;
If the triples $(e_{I_1},e_{J_1},e_{K_1})$ and
$(e_{I_2},e_{J_2},e_{K_2})$ share no common edge, the tensor-product
structure~(\ref{multiedge}) of the multi-edge coherent state
\begin{equation}\label{multiedge}
|\psi^t_{\Gamma,\{g\}}\rangle
:=|\psi^t_{g_{e_1}}\rangle
  \otimes|\psi^t_{g_{e_2}}\rangle
  \otimes\cdots\otimes|\psi^t_{g_{e_n}}\rangle
\end{equation}
implies that Eq.~(\ref{psum123}) factorizes into single-edge matrix
elements.
Since both numerator and denominator contain exactly the same
collection of single-edge factors for each edge (with overlap functions
for unacted edges cancelling identically), the ratio evaluates to~$1$
exactly:
\begin{equation}
\begin{split}
&\frac{
  \langle\psi^t_{g_{e_{I_1}}}|\widehat{p}^{i_1}_{e_{I_1}}|
  \psi^t_{g'_{e_{I_1}}}\rangle\;
  \langle\psi^t_{g_{e_{J_1}}}|\widehat{p}^{j_1}_{e_{J_1}}|
  \psi^t_{g'_{e_{J_1}}}\rangle\;
  \langle\psi^t_{g_{e_{K_1}}}|\widehat{p}^{k_1}_{e_{K_1}}|
  \psi^t_{g'_{e_{K_1}}}\rangle\;
  \langle\psi^t_{g_{e_{I_2}}}|\widehat{p}^{i_2}_{e_{I_2}}|
  \psi^t_{g'_{e_{I_2}}}\rangle\;
  \langle\psi^t_{g_{e_{J_2}}}|\widehat{p}^{j_2}_{e_{J_2}}|
  \psi^t_{g'_{e_{J_2}}}\rangle\;
  \langle\psi^t_{g_{e_{K_2}}}|\widehat{p}^{k_2}_{e_{K_2}}|
  \psi^t_{g'_{e_{K_2}}}\rangle}
{
  \langle\psi^t_{g_{e_{I_1}}}|\widehat{p}^{i_1}_{e_{I_1}}|
  \psi^t_{g'_{e_{I_1}}}\rangle\;
  \langle\psi^t_{g_{e_{J_1}}}|\widehat{p}^{j_1}_{e_{J_1}}|
  \psi^t_{g'_{e_{J_1}}}\rangle\;
  \langle\psi^t_{g_{e_{K_1}}}|\widehat{p}^{k_1}_{e_{K_1}}|
  \psi^t_{g'_{e_{K_1}}}\rangle\;
  \langle\psi^t_{g_{e_{I_2}}}|\widehat{p}^{i_2}_{e_{I_2}}|
  \psi^t_{g'_{e_{I_2}}}\rangle\;
  \langle\psi^t_{g_{e_{J_2}}}|\widehat{p}^{j_2}_{e_{J_2}}|
  \psi^t_{g'_{e_{J_2}}}\rangle\;
  \langle\psi^t_{g_{e_{K_2}}}|\widehat{p}^{k_2}_{e_{K_2}}|
  \psi^t_{g'_{e_{K_2}}}\rangle}\\
&\quad=1.
\end{split}
\end{equation}

\textbf{Case~2: overlapping edge sets.}\;
Suppose, for instance, $I_1=I_2$ while $J_1\neq J_2$ and
$K_1\neq K_2$.
The tensor-product factorization now yields
\begin{equation}
\begin{split}
&\frac{
  \langle\psi^t_{g_{e_{I_1}}}|
    \widehat{p}^{i_1}_{e_{I_1}}\widehat{p}^{i_2}_{e_{I_1}}
  |\psi^t_{g'_{e_{I_1}}}\rangle\;
  \langle\psi^t_{g_{e_{I_1}}}|\psi^t_{g'_{e_{I_1}}}\rangle\;
  \langle\psi^t_{g_{e_{J_1}}}|\widehat{p}^{j_1}_{e_{J_1}}|
  \psi^t_{g'_{e_{J_1}}}\rangle\;
  \langle\psi^t_{g_{e_{K_1}}}|\widehat{p}^{k_1}_{e_{K_1}}|
  \psi^t_{g'_{e_{K_1}}}\rangle\;
  \langle\psi^t_{g_{e_{J_2}}}|\widehat{p}^{j_2}_{e_{J_2}}|
  \psi^t_{g'_{e_{J_2}}}\rangle\;
  \langle\psi^t_{g_{e_{K_2}}}|\widehat{p}^{k_2}_{e_{K_2}}|
  \psi^t_{g'_{e_{K_2}}}\rangle
}
{
  \langle\psi^t_{g_{e_{I_1}}}|\widehat{p}^{i_1}_{e_{I_1}}|
  \psi^t_{g'_{e_{I_1}}}\rangle\;
  \langle\psi^t_{g_{e_{I_1}}}|\widehat{p}^{i_2}_{e_{I_1}}|
  \psi^t_{g'_{e_{I_1}}}\rangle\;
  \langle\psi^t_{g_{e_{J_1}}}|\widehat{p}^{j_1}_{e_{J_1}}|
  \psi^t_{g'_{e_{J_1}}}\rangle\;
  \langle\psi^t_{g_{e_{K_1}}}|\widehat{p}^{k_1}_{e_{K_1}}|
  \psi^t_{g'_{e_{K_1}}}\rangle\;
  \langle\psi^t_{g_{e_{J_2}}}|\widehat{p}^{j_2}_{e_{J_2}}|
  \psi^t_{g'_{e_{J_2}}}\rangle\;
  \langle\psi^t_{g_{e_{K_2}}}|\widehat{p}^{k_2}_{e_{K_2}}|
  \psi^t_{g'_{e_{K_2}}}\rangle
}\\[6pt]
&\quad=
\frac{\langle\psi^t_{g_{e_{I_1}}}|
        \widehat{p}^{i_1}_{e_{I_1}}\widehat{p}^{i_2}_{e_{I_1}}
      |\psi^t_{g'_{e_{I_1}}}\rangle\;
      \langle\psi^t_{g_{e_{I_1}}}|\psi^t_{g'_{e_{I_1}}}\rangle}
     {\langle\psi^t_{g_{e_{I_1}}}|\widehat{p}^{i_1}_{e_{I_1}}|
      \psi^t_{g'_{e_{I_1}}}\rangle\;
      \langle\psi^t_{g_{e_{I_1}}}|\widehat{p}^{i_2}_{e_{I_1}}|
      \psi^t_{g'_{e_{I_1}}}\rangle}
=1+\mathcal{O}(t),
\end{split}
\end{equation}
where the last step follows from Eq.~(\ref{respe}).
All remaining overlap patterns reduce analogously, giving
$1+\mathcal{O}(t)$ in every case.
By the algebraic identity~(\ref{sum_alphak}), the full ratio
(\ref{multiedgeQ}) satisfies
\begin{equation}\label{resQeF}
\frac{\langle\psi^t_{\Gamma,\{g\}}|\psi^t_{\Gamma,\{g'\}}\rangle\;
      \langle\psi^t_{\Gamma,\{g\}}|\widehat{Q}_v^2|
      \psi^t_{\Gamma,\{g'\}}\rangle}
     {\langle\psi^t_{\Gamma,\{g\}}|\widehat{Q}_v|
      \psi^t_{\Gamma,\{g'\}}\rangle^2}
=1+\mathcal{O}(t),
\end{equation}
which establishes condition~(1).

\subsubsection*{Verification of condition~(2)}

To verify condition~(2), we examine the effective action governing the
saddle-point expansion of Eq.~(\ref{intF}) in the LQG setting.
The relevant action reads
\begin{equation}\label{LQGaction}
S[g_c,g,g']
:=\frac{1}{t}\left(K(g,g_c)+K(g_c,g')\right)
 =\frac{1}{t}\left(
   m(g,g_c)^2+m(g_c,g')^2
   -\tfrac{1}{2}p(g)^2-\tfrac{1}{2}p(g')^2-p(g_c)^2
 \right),
\end{equation}
where $K$ represents the K\"ahler potential and $p(g):=m(g,g)$.
Geometrically, on the cotangent bundle
$T^*SU(2)\cong SL(2,\mathbb{C})$, the function $K(a,b)$ is a
K\"ahler potential representing a squared geodesic distance.
Variation of $S$ constrains the critical point $g_c$ to lie on the
classical geodesic connecting $g$ and $g'$ in phase space.

Since the base manifold $SU(2)$ is compact, infinitely many geodesics
connect any pair of points.
However, the exponential weight $e^{-S/t}$ ensures that only the
\emph{shortest} geodesic contributes at leading order;
all longer geodesics are exponentially suppressed.
Moreover, within the coherent-state path integral each time slice is thin, so that $g$ and $g'$ lie in a sufficiently small neighborhood in phase space, one can then  prove there exists a single, isolated,
non-degenerate critical point \cite{Han:2019vpw}.

\subsubsection*{Extension to gauge-invariant coherent states}

The expansion extends to the gauge-invariant sector.
The gauge-invariant coherent state is
\begin{equation}
|\Psi^t_{\Gamma,\{g\}}\rangle
=P_{\{h\}}\,\psi^t_{\Gamma,\{g\}}(\{h\}),
\end{equation}
where $P_{\{h\}}$ is the projection operator implementing the $SU(2)$
group averaging of Eq.~(\ref{groupA}).
The key algebraic properties are
\begin{equation}\label{P_prop}
\begin{split}
&\langle P_{\{h\}}\psi^t_{\Gamma,\{g\}},\,
         P_{\{h\}}\psi^t_{\Gamma,\{g'\}}\rangle
=\langle\psi^t_{\Gamma,\{g\}},\,
         P_{\{h\}}\psi^t_{\Gamma,\{g'\}}\rangle,\\
&[P_{\{h\}},\widehat{Q}_v]=0,
\qquad
 [P_{\{h\}},\widehat{Q}_v^n]=0,
\end{split}
\end{equation}
the commutativity following from the gauge invariance of
$\widehat{Q}_v$.
Using these properties, the variance of $\widehat{Q}_v^2$ in the
gauge-invariant representation evaluates as
\begin{equation}
\begin{split}
&\langle\Psi^t_{\Gamma,\{g\}}|
\left[
  \widehat{Q}_v^2
  -\left(
    \frac{\langle\Psi^t_{\Gamma,\{g\}}|\widehat{Q}_v|
          \Psi^t_{\Gamma,\{g'\}}\rangle}
         {\langle\Psi^t_{\Gamma,\{g\}}|
          \Psi^t_{\Gamma,\{g'\}}\rangle}
  \right)^{\!2}
\right]
|\Psi^t_{\Gamma,\{g'\}}\rangle\\
={}&
\int dh\;\langle\psi^t_{\Gamma,\{g\}}|
\left[
  \widehat{Q}_v^2
  -\left(
    \frac{\langle\Psi^t_{\Gamma,\{g\}}|\widehat{Q}_v|
          \Psi^t_{\Gamma,\{g'\}}\rangle}
         {\langle\Psi^t_{\Gamma,\{g\}}|
          \Psi^t_{\Gamma,\{g'\}}\rangle}
  \right)^{\!2}
\right]
|\psi^t_{\Gamma,\alpha_{\{h\}}(\{g'\})}\rangle\\
\sim{}&
\operatorname{Pol}[\{g\},\{g'\}]\;
\langle\psi^t_{\Gamma,\{g\}}|
\left[
  \widehat{Q}_v^2
  -\left(
    \frac{\langle\psi^t_{\Gamma,\{g\}}|\widehat{Q}_v|
          \psi^t_{\Gamma,\alpha_{\{h_c\}}(\{g'\})}\rangle}
         {\langle\psi^t_{\Gamma,\{g\}}|
          \psi^t_{\Gamma,\alpha_{\{h_c\}}(\{g'\})}\rangle}
  \right)^{\!2}
\right]
|\psi^t_{\Gamma,\alpha_{\{h_c\}}(\{g'\})}\rangle
\left(1+\mathcal{O}(t)\right)\\
={}&
t\,\operatorname{Pol}[\{g\},\{g'\}]\;
\langle\psi^t_{\Gamma,\{g\}}|
\psi^t_{\Gamma,\alpha_{\{h_c\}}(\{g'\})}\rangle
\left(1+\mathcal{O}(t)\right)\\
={}&
t\,\operatorname{Pol}[\{g\},\{g'\}]\;
\langle\Psi^t_{\Gamma,\{g\}}|
\Psi^t_{\Gamma,\{g'\}}\rangle
+\mathcal{O}(t^2),
\end{split}
\end{equation}
where in the second step we performed saddle point approximation. The third step is due to Eq. (\ref{resQe}). And the last step is achieved by taking into account of the saddle point approximation of $\langle\Psi^t_{\Gamma,\{g\}}|
\Psi^t_{\Gamma,\{g'\}}\rangle$. $\alpha_{\{h\}}$ denotes the gauge transformation in the integrand of
Eq.~(\ref{groupA}), and $h_c$ is the dominant saddle point of the
gauge-group integration and $\operatorname{Pol}[\{g\},\{g'\}]$ is a polynomial in the group elements independent of $t$. 

By the same token, the higher-order relation analogous to
Eq.~(\ref{Q4_new}) holds:
\begin{equation}
\begin{split}
&\langle\Psi^t_{\Gamma,\{g\}}|
\left(
  \widehat{Q}_v^4
  -\frac{2\langle\Psi^t_{\Gamma,\{g\}}|\widehat{Q}_v|
         \Psi^t_{\Gamma,\{g'\}}\rangle^2\;\widehat{Q}_v^2}
        {\langle\Psi^t_{\Gamma,\{g\}}|
         \Psi^t_{\Gamma,\{g'\}}\rangle^2}
  +\frac{\langle\Psi^t_{\Gamma,\{g\}}|\widehat{Q}_v|
         \Psi^t_{\Gamma,\{g'\}}\rangle^4}
        {\langle\Psi^t_{\Gamma,\{g\}}|
         \Psi^t_{\Gamma,\{g'\}}\rangle^4}
\right)
|\Psi^t_{\Gamma,\{g'\}}\rangle\\
={}&
t\,\operatorname{Pol}[\{g\},\{g'\}]\;
\langle\Psi^t_{\Gamma,\{g\}}|
\Psi^t_{\Gamma,\{g'\}}\rangle
+\mathcal{O}(t^2),
\end{split}
\end{equation}
and the inductive assumption~(\ref{asumpO}) together with the
splitting~(\ref{ON1}) in lemma \ref{theoremAsymp} carry over to the gauge-invariant sector by
virtue of Eq.~(\ref{P_prop}).

\subsubsection*{Holonomy operators}

The quantum holonomy $\widehat{h}_{AB}$ can be expressed in terms
of the complexified group element and the flux operator
\cite{Thiemann:2000ca}:
\begin{equation}\label{hquantum}
\widehat{h}_{AB}
=e^{-3t/8}\left(e^{i\widehat{p}_j\tau_j/2}\,\widehat{g}\right)_{AB}.
\end{equation}
For monomials of holonomy operators,
\begin{equation}\label{hexpand}
\langle\psi^t_g|
  \widehat{h}_{A_1 B_1}\cdots\widehat{h}_{A_n B_n}
|\psi^t_{g'}\rangle
=e^{-3t/8}
\sum_C g'_{C B_n}\,
\langle\psi^t_g|
  \widehat{h}_{A_1 B_1}\cdots\widehat{h}_{A_{n-1}B_{n-1}}\,
  \left(e^{i\widehat{p}_j\tau_j/2}\right)_{A_n C}
|\psi^t_{g'}\rangle.
\end{equation}
The commutation relation
$[\widehat{p}_j,\widehat{h}_{AB}]
=\tfrac{it}{2}(\tau_j\widehat{h})_{AB}$
gives
\begin{equation}
\left[
  \tfrac{i}{2}\textstyle\sum_j\widehat{p}_j\tau_j,\;
  \widehat{h}_{A_{n-1}B_{n-1}}
\right]
=\tfrac{i}{2}\sum_j
 [\widehat{p}_j\tau_j,\widehat{h}_{A_{n-1}B_{n-1}}]
=-\tfrac{t}{4}\left(\textstyle\sum_j\tau_j^2\right)
 \widehat{h}_{A_{n-1}B_{n-1}}
=\tfrac{3t}{4}\,\widehat{h}_{A_{n-1}B_{n-1}},
\end{equation}
where we used $\sum_j\tau_j^2=-3\cdot\mathbf{1}$.
Exponentiating via the Baker--Campbell--Hausdorff formula,
\begin{equation}
\begin{split}
\left(e^{i\widehat{p}_j\tau_j/2}\right)_{A_n C}\,
\widehat{h}_{A_{n-1}B_{n-1}}\,
\left(e^{-i\widehat{p}_j\tau_j/2}\right)_{A_n C}
&=\widehat{h}_{A_{n-1}B_{n-1}}
  +\tfrac{3t}{4}\widehat{h}_{A_{n-1}B_{n-1}}
  +\tfrac{1}{2!}\left(\tfrac{3t}{4}\right)^2
   \widehat{h}_{A_{n-1}B_{n-1}}+\cdots\\
&=e^{3t/4}\,\widehat{h}_{A_{n-1}B_{n-1}},
\end{split}
\end{equation}
so that
\begin{equation}
\left(e^{i\widehat{p}_j\tau_j/2}\right)_{A_n C}\,
\widehat{h}_{A_{n-1}B_{n-1}}
=e^{3t/4}\,\widehat{h}_{A_{n-1}B_{n-1}}\,
\left(e^{i\widehat{p}_j\tau_j/2}\right)_{A_n C}.
\end{equation}
Iterating, one obtains the final factorization
\begin{equation}\label{hexpandF}
\langle\psi^t_g|
  \widehat{h}_{A_1 B_1}\cdots\widehat{h}_{A_n B_n}
|\psi^t_{g'}\rangle
=e^{-3n^2 t/8}
\sum_{C_1,\dots,C_n}
g'_{C_1 B_1}\cdots g'_{C_n B_n}\;
\langle\psi^t_g|
  \left(e^{i\widehat{p}_j\tau_j/2}\right)_{A_1 C_1}
  \cdots
  \left(e^{i\widehat{p}_j\tau_j/2}\right)_{A_n C_n}
|\psi^t_{g'}\rangle.
\end{equation}
Consequently, polynomials of the holonomy operator can in
principle be computed following the methods of
\cite{Thiemann:2000bx}.

\subsection{Numerical Comparison}\label{sec:LQGN}

In this section we compare the analytical expansion formula proposed
in this work with the old expansion in LQG, as well as with numerical
results obtained using the algorithm developed
in~\cite{Liliu:202603,LQG-Volume}.
We first introduce the numerical algorithm and the specific setup used
in the computation.
We then compare the expansion coefficients of the two formulas for a
gauge-variant 3-bridges graph and examine how their discrepancy scales
with the relative angle~$\theta$ between the coherent states.
Finally, we validate the accuracy of the proposed expansion by
contrasting these analytical predictions with numerical data.

\subsubsection{Setup}

The underlying program is implemented in the Julia environment and
makes extensive use of parallel computing (typically 40--80 CPU cores).
The matrix elements of the volume operator acting on (normalized)
gauge-invariant coherent states are computed as
\begin{equation}\label{sum111}
    \begin{split}
        \langle \Psi^t_{\Gamma,\{g\}}|\widehat{V}_{v}|
        \Psi^t_{\Gamma,\{g'\}} \rangle
        =\sum\limits_{\{\vec{j}\}}\Bigg(
        \sum\limits_{\lambda,\lambda'}
        \sum\limits_{\{\vec{\iota}\},\{\vec{\iota}'\}}
        &\langle \Psi^t_{\Gamma,\{g\}}|
        \Gamma,\{\vec{\iota}\},\{\vec{j}\}\rangle
        \langle\Gamma,\{\vec{\iota}\},\{\vec{j}\}|\lambda\rangle
        \langle\lambda|\widehat{V}_v|\lambda'\rangle\\
        &\times\langle\lambda'|\Gamma,\{\vec{\iota}'\},\{\vec{j}\}\rangle
        \langle\Gamma,\{\vec{\iota}'\},\{\vec{j}\}|
        \Psi^t_{\Gamma,\{g'\}} \rangle\Bigg),
    \end{split}
\end{equation}
while the corresponding expression for gauge-variant coherent states
reads
\begin{equation}\label{sum333}
    \begin{split}
        \langle\psi^t_{\Gamma,\{g\}}|\widehat{V}_v|
        \psi^t_{\Gamma,\{g'\}} \rangle
        =\sum\limits_{\{\vec{j}\},\{M\},\{M'\}}\Bigg(
        \sum\limits_{\lambda,\lambda'}
        \sum\limits_{\{\vec{\iota}\},\{\vec{\iota}'\}}
        &\langle \psi^t_{\Gamma,\{g\}}|
        \Gamma,\{\vec{\iota}\},\{\vec{j}\},\{M\}\rangle
        \langle\Gamma,\{\vec{\iota}\},\{\vec{j}\},\{M\}|
        \lambda\rangle\\
        &\times\langle\lambda|\widehat{V}_v|\lambda'\rangle
        \langle\lambda'|\{\vec{\iota}'\},\{\vec{j}\},\{M'\}\rangle
        \langle\{\vec{\iota}'\},\{\vec{j}\},\{M'\}|
        \psi^t_{\Gamma,\{g'\}} \rangle\Bigg),
    \end{split}
\end{equation}
where $|\lambda\rangle$ and $|\lambda'\rangle$ are eigenvectors of
$\widehat{V}_v$, and
$\langle\lambda'|\{\vec{\iota}'\},\{\vec{j}\},\{M'\}\rangle$ is the
transformation matrix between the eigenbasis and the spin-network
basis.
Since $\langle\lambda|\widehat{V}_v|\lambda'\rangle$ is diagonal, the
square root of~$\widehat{Q}_v$ can be computed straightforwardly.
The expectation value of~$\widehat{V}$ on
$|\tilde{\Psi}^t_{\Gamma,\{g\}}\rangle$ is obtained by setting
$\{g'\}=\{g\}$ in the above equations.
Here $V_{v}=\frac{1}{8}\left(Q_{v}^2\right)^{1/4}$, where the
coefficient~$\frac{1}{8}$ ensures the correct correspondence with the
classical volume in LQG (see, e.g.,
\cite{Han:2019vpw}~Eq.~(7.1)).
The computation of the volume-operator matrix elements on
spin-network states is detailed
in~\cite{Brunnemann:2004xi}, and the transformation matrix between
the spin-network and coherent-state representations can be found
in~\cite{PhysRevD.92.104023}.

\begin{figure}[h!]
 \centering
 \includegraphics[height=5cm]{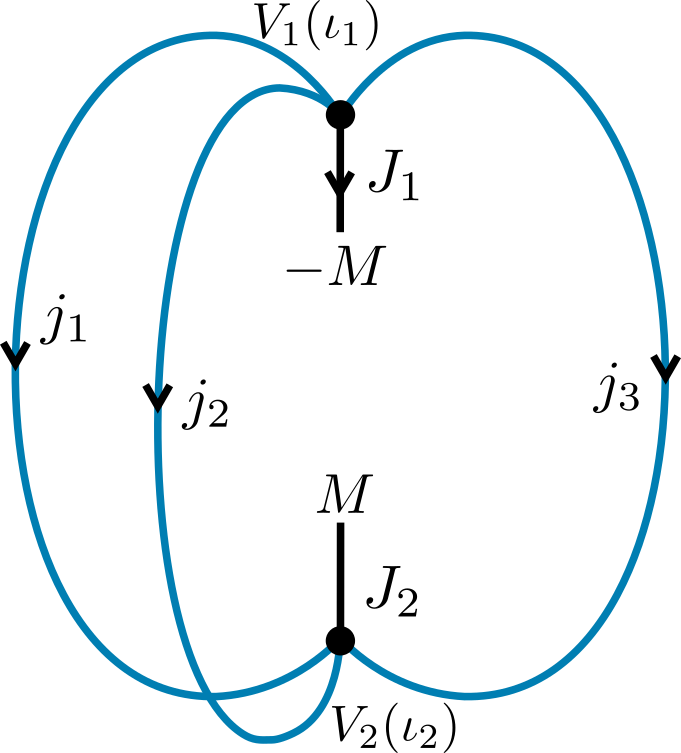}
 \caption{Spin-network structure of a gauge-variant
          3-bridges graph.}\label{fig3-1}
\end{figure}

We now describe the specific setup used to compare the semiclassical
expansions with the numerical results.
A gauge-variant 3-bridges graph is selected (see
Fig.~\ref{fig3-1}, with vertices~$V_1$ and~$V_2$).
Two sets of $SL(2,\mathbb{C})$ elements are assigned to the three
edges, denoted~$\{g\}$ and~$\{g'\}$.
The set~$\{g'\}$ is defined by $\vec{z}'_I=0$,
$\vec{p}'_1=(6\cos\theta,\,6\sin\theta,\,0)$,
$\vec{p}'_2=(0,\,6,\,0)$, and $\vec{p}'_3=(0,\,0,\,6)$;
the set~$\{g\}$ is defined by $\vec{z}_I=0$,
$\vec{p}_1=(6,\,0,\,0)$, $\vec{p}_2=(0,\,6,\,0)$, and
$\vec{p}_3=(0,\,0,\,6)$.
With these data one constructs two Thiemann coherent states
$\psi^t_{\Gamma,\{g\}}$ and $\psi^t_{\Gamma,\{g'\}}$, and the
matrix elements of~$\widehat{V}_{V_1}$ can then be evaluated both
analytically and numerically.

\subsubsection{Comparison of the two expansions}

We now test the two operator-expansion formulas: the ``old'' expansion
Eq.~(\ref{Vexpand}), originally introduced
in~\cite{Giesel:2006um}, and the ``new'' expansion
Eq.~(\ref{VexpandNew}), proposed in the present work.
Both formulas yield the matrix element of~$\widehat{V}_{V_1}$ for
gauge-variant 3-bridges.

Table~\ref{tab3mat} compares the old and new expansion formulas for
the quantity
$\bigl|\langle\psi^t_{\Gamma,\{g\}}|\widehat{V}_{V_1}|
\psi^t_{\Gamma,\{g'\}}\rangle\big/
\langle\psi^t_{\Gamma,\{g\}}|\psi^t_{\Gamma,\{g'\}}\rangle\bigr|$
up to order~$t^2$.
Here~$\theta$ denotes the angle between the label data of the first
edge in $|\psi^t_{\Gamma,\{g\}}\rangle$ and
$|\psi^t_{\Gamma,\{g'\}}\rangle$.
As the table shows, the two expansions already differ at zeroth order
and the discrepancy grows with~$\theta$, i.e.\ as the two coherent
states deviate further from each other.
The difference becomes appreciable only for large~$\theta$, where the
directions of the first edge in the two states become nearly
antiparallel.
Moreover, the difference is more pronounced for the $t^2$~coefficient
than for the $t^1$~coefficient, which in turn deviates more than the
$t^0$~term.
This indicates that higher-order coefficients are increasingly
sensitive to large~$\theta$; we expect terms beyond~$t^2$ to exhibit
the same trend.

\begin{table}[h!]
 \begin{center}
\footnotesize
\renewcommand{\arraystretch}{1.5}
\caption{Comparison of the old and new expansion formulas up to
$\mathcal{O}(t^2)$ for
$\bigl|\langle\psi^t_{\Gamma,\{g\}}|\widehat{V}_{V_1}|
\psi^t_{\Gamma,\{g'\}}\rangle\big/
\langle\psi^t_{\Gamma,\{g\}}|
\psi^t_{\Gamma,\{g'\}}\rangle\bigr|$.}\label{tab3mat}
\begin{tabular}{ |c||c|c| }
 \hline
 $\theta$ & Old expansion & New expansion\\
 \hline
 $30^{\circ}$
   & $1.8262 - 0.200433\, t - 0.00175893\, t^2$
   & $1.8262 - 0.200433\, t - 0.00178565\, t^2$\\
 \hline
 $60^{\circ}$
   & $1.79132 - 0.198826\, t - 0.00147641\, t^2$
   & $1.79132 - 0.198834\, t - 0.00184545\, t^2$\\
 \hline
 $90^{\circ}$
   & $1.72432 - 0.195579\, t - 0.000493382\, t^2$
   & $1.72431 - 0.19581\, t - 0.00199745\, t^2$\\
 \hline
 $120^{\circ}$
   & $1.60216 - 0.187737\, t + 0.00170764\, t^2$
   & $1.6018 - 0.190451\, t - 0.00244614\, t^2$\\
 \hline
 $150^{\circ}$
   & $1.34692 - 0.156606\, t + 0.00743893\, t^2$
   & $1.33168 - 0.179586\, t - 0.00496832\, t^2$\\
 \hline
\end{tabular}
\renewcommand{\arraystretch}{1}
\end{center}
\end{table}

To further validate the new formula, we have numerically evaluated the
quantity
$\bigl|\langle\psi^t_{\Gamma,\{g\}}|\widehat{V}_{V_1}|
\psi^t_{\Gamma,\{g'\}}\rangle\big/
\langle\psi^t_{\Gamma,\{g\}}|\psi^t_{\Gamma,\{g'\}}\rangle\bigr|$
for~$\theta=120^\circ$.
Fig.~\ref{V4matdiff1}\,(a) compares the numerically determined matrix
elements of~$\widehat{V}_{V_1}$ with both expansion formulas.
Near $t=0.65$, both the old and new expansions begin to converge
toward the numerical result, illustrating that the two expansions
closely coincide in the small-$t$ regime even at large~$\theta$.

A magnified view is given in Fig.~\ref{V4matdiff1}\,(b), restricted
to $0.6\leq t\leq 1$.
This confirms that the expansion proposed in this work,
Eq.~(\ref{VexpandNew}), is indeed more accurate.
Fig.~\ref{V4matdiff2} contrasts the relative discrepancies between
the numerical data and both expansion formulas.
The figure shows that the relative difference is noticeably smaller
for the new expansion when $t$ is small.
However, the numerical data still exhibit oscillatory behavior in the
interval $0.65\leq t\leq 3$, so the full convergence process has not
yet been resolved; doing so would likely require data reaching down
to~$t\leq 0.25$. This is because, as Table~\ref{tab3mat} shows, the two expansions
already differ at zeroth order (i.e.\ at $t=0$).
Consequently, their relative validity can be conclusively tested in
the small-$t$ regime.

In summary, Table~\ref{tab3mat} and Fig.~\ref{V4matdiff2} show that
when the two coherent states are close, both the old and new
expansions produce virtually the same approximation.
As the two states diverge, the results begin to differ
significantly---even at leading order---clearly exposing the
limitations of the old expansion.
With future improvements in algorithms and increased computational
resources, more precise results can be expected.

\begin{figure}[h!]
 \centering
 \begin{subfigure}[b]{0.45\textwidth}
   \centering
   \includegraphics[height=5cm]{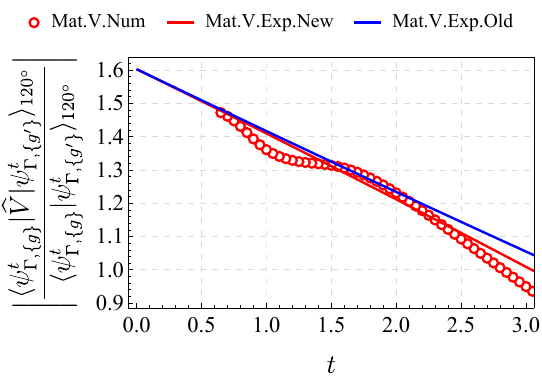}
   \caption{Matrix element}
 \end{subfigure}
 \begin{subfigure}[b]{0.45\textwidth}
   \centering
   \includegraphics[height=5cm]{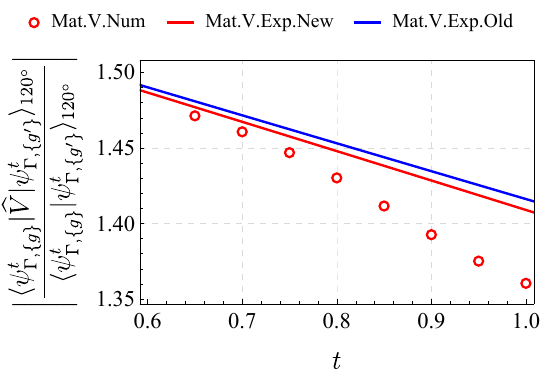}
   \caption{Zoom-in for $0.6\leq t\leq 1.0$}
 \end{subfigure}
 \caption{Berezin Symbol of $\widehat{V}_{V_1}$ for
          $\theta=120^{\circ}$.}\label{V4matdiff1}
\end{figure}

\begin{figure}[h!]
 \centering
 \includegraphics[height=5cm]{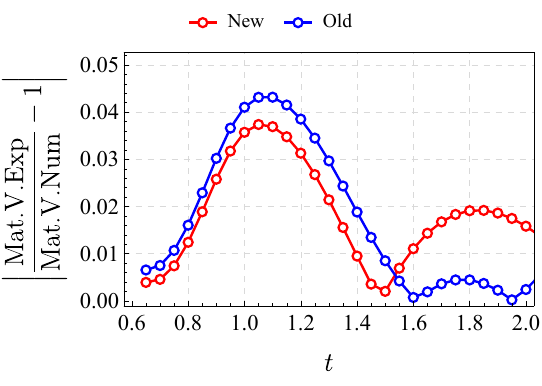}
 \caption{Relative discrepancy between the numerical data and the
          old/new expansion formulas for
          $\theta=120^\circ$.}\label{V4matdiff2}
\end{figure}

\section{Discussion}\label{sectionV}

In the present work, a novel expansion formula has been developed for the purpose of computing the matrix elements of non-polynomial operators within the coherent state representation of gauge theories. By utilizing saddle point approximation techniques in conjunction with the operator-algebraic properties of the matrix elements, we have successfully proven the asymptotic expansion formula and investigated the scope and the necessary requirements for its physical application. Based on the mechanism of our proof, it is highly universal and can be applied to most analytic non-polynomial functions $f(\widehat{A})$.

As a concrete and significant physical application, we have demonstrated that this general expansion formula correctly executes the expansion of non-polynomial functions of the flux operator—and consequently, the physical Hamiltonian operator—within the framework of canonical quantum gravity. By integrating our expansion formula with the established Thiemann coherent state techniques, the complex, non-perturbative Hamiltonian operator in LQG can be systematically expanded and evaluated. To demonstrate that the derived operator expansion captures the lower-order quantum fluctuations with sufficient accuracy, we have compared our novel expansion against both the existing Giesel-Thiemann expansion formula \cite{Giesel:2006um} and numerical results obtained independently in the deep quantum regime. The numerical evaluations are acquired by computing the matrix elements in the spin-network representation and projecting them onto the coherent state basis. The outcomes evidently show that, while the earlier expectation-value-based expansion is well-suited for formally defining the theory in the continuous-time limit (where the initial and final states heavily overlap), it begins to deviate from the true semiclassical phase once the discrepancy between the two coherent states increases (as is typical in finite-step lattice computations), primarily due to the lack of the full holomorphic structure. Conversely, our newly proposed expansion—by anchoring the quantum fluctuations to the true off-diagonal geometric saddle points connecting the states—maintains structural precision and correctly reproduces the numerical matrix elements, thereby illustrating the practical advantage of our formalism. 

The theoretical generality of our work possesses profound implications for quantum field theories. Although our framework operates on well-defined quantum operators initially constructed via standard Dirac quantization, the asymptotic evaluation of their Berezin symbols developed herein mirrors the algebraic properties of deformation quantization. By evaluating the operator action via the generalized gauge group coherent states (such as the Hall and Thiemann states), our expansion mechanism systematically generates the quantum corrections that could be analogous to the non-commutative star product. In principle, this analytical control over non-polynomial operators provides an useful tool for investigating non-perturbative fixed points and executing the systematic cancellation of quantum anomalies. Furthermore, as the coherent state path integral formulation depends on the evaluation of off-diagonal matrix elements on the lattice, our expansion formalism offers a possible route for the construction of lattice gauge path integrals and their subsequent analytical continuations.

Looking forward, the generalized nature of this expansion framework opens several promising avenues for future research. A first issue is to understand how the present operator expansion should be modified in degenerate sectors of the phase space, where the off-diagonal symbols may vanish and the current non-vanishing assumptions break down. Such configurations may be physically relevant in quantum gravity, and treating them consistently may require new regularization schemes or alternative expansion variables. Beyond this, because our methodology essentially relies on the analytic continuation of matrix elements across the complexified phase space, it naturally facilitates the study of instanton physics. Instantons and other tunneling phenomena require the evaluation of path integrals around multiple critical points within topologically non-trivial phase spaces; understanding how the present operator expansion adapts to quantum interference between distinct classical geodesics will be a subject of considerable interest. Additionally, we intend to explore the analytic continuation of physical parameters—such as the transition from real time to Euclidean time which defines the thermal partition function, or the further complexification of the fundamental gauge group. Such investigations are required for the proper definition of perturbation theory in strongly fluctuating, non-perturbative quantum systems. Finally, explicitly applying this fully expanded Hamiltonian operator to evaluate the transition amplitudes and the resulting renormalization flows and smooth continuum dynamics within the complete gauge coherent state path integral remains our ultimate objective.


\bibliographystyle{jhep}
\bibliography{sample}

\end{document}